%% file: UCATE_NeurIPS.tex
\documentclass{article}

\PassOptionsToPackage{}{natbib}

\usepackage[main,final]{neurips_2026}

\usepackage[utf8]{inputenc} 
\usepackage[T1]{fontenc}    
\usepackage{hyperref}       
\hypersetup{hypertexnames=false} 
\usepackage{url}            
\usepackage{booktabs}       
\usepackage{amsfonts}       
\usepackage{nicefrac}       
\usepackage{microtype}      
\usepackage{xcolor}         
\newif\ifshowrevisions
\showrevisionsfalse
\newcommand{\revisioncolor}{\ifshowrevisions\color{blue}\fi}
\newcommand{\rev}[1]{{\revisioncolor #1}}
\newenvironment{revised}{\begingroup\revisioncolor}{\endgroup}
\ifshowrevisions\else
\renewcommand{\answerYes}[1][]{[Yes]#1}
\renewcommand{\answerNo}[1][]{[No]#1}
\renewcommand{\answerNA}[1][]{[N/A]#1}
\fi

\usepackage{amsmath,amssymb,amsthm,mathtools}
\usepackage{algorithm}
\usepackage{algorithmic}
\usepackage{bm}
\usepackage[capitalize,noabbrev]{cleveref}
\usepackage{enumitem}
\usepackage{etoc}
\usepackage{graphicx}
\usepackage{mathrsfs}
\usepackage{tcolorbox}
\usepackage{kotex}
\theoremstyle{plain}
\newtheorem{theorem}{Theorem}
\newtheorem{lemma}[theorem]{Lemma}
\newtheorem{proposition}[theorem]{Proposition}
\newtheorem{corollary}[theorem]{Corollary}

\theoremstyle{definition}
\newtheorem{definition}[theorem]{Definition}
\newtheorem{assumption}[theorem]{Assumption}
\newtheorem{remark}[theorem]{Remark}

\etocdepthtag.toc{mtchapter}
\etocsettagdepth{mtchapter}{subsection}
\etocsettagdepth{mtappendix}{none}

\input{macros}
\pgfplotsset{compat=1.18}
\title{Structure-Adaptive Estimation of Heterogeneous Treatment Effects with Kernel Methods}

\author{%
  \rev{Seok-Jin Kim} \\
  \rev{Columbia University} \\
  \rev{New York, NY, USA} \\
  \rev{\texttt{seok-jin.kim@columbia.edu}} \\
}

\begin{document}

\maketitle

\begin{abstract}%
We study conditional average treatment effect (CATE) estimation from observational data in regimes with potentially weak overlap, when the nuisance response functions lie in a reproducing kernel Hilbert space (RKHS), but the contrast function is structurally simpler.
We consider three forms of such simplicity: (i) membership in a lower-complexity RKHS with faster eigenvalue decay, (ii) a source condition relative to the nuisance kernel, or (iii) dependence on a known low-dimensional covariate representation.
We develop a unified two-stage kernel ridge regression procedure that remains valid under weak overlap and attains minimax rates determined by the complexity of the contrast function.
\rev{The procedure avoids propensity-score estimation and separate nuisance convergence-rate assumptions.}
We also show that a simple model-selection step over candidate contrast spaces and regularization levels satisfies an oracle inequality, enabling adaptation to unknown CATE regularity.
\end{abstract}

\section{Introduction}

Estimating heterogeneous treatment effects is a central challenge in causal inference. In particular, \rev{conditional average treatment effects (CATEs)} are fundamental for personalized decision-making in domains ranging from precision medicine to economics \citep{kunzel2019metalearners,kennedy2023towards}. While the Average Treatment Effect (ATE) provides a population-level summary, it can obscure substantial heterogeneity across individuals. 
A recurring feature in this literature is the structural asymmetry between the treatment effect of interest and the nuisance components, such as the baseline response functions. Empirically, nuisance functions may be highly complex, non-smooth, or high-dimensional even when the CATE itself is smooth, sparse, or nearly constant \citep{kennedy2023towards,kennedy2024minimax,kato2023cate}. A central theoretical question is whether one can exploit this simpler contrast structure to obtain convergence rates faster than those dictated by nuisance learning, while still achieving minimax-optimal performance.

A standard route to exploiting such structure is \rev{doubly robust (DR) estimation} or double machine learning (DML). These methods typically yield an oracle term together with a product of nuisance estimation errors, but the norms involved can be quite demanding in CATE problems. For ATE, the product is often controlled in \(L^2\), whereas for CATE it may require stronger \(L^4\)-type control, and for pointwise guarantees it can require pointwise accuracy of both the outcome regression and the propensity score. As a result, obtaining sharp adaptive guarantees is substantially harder in the heterogeneous-treatment setting.
A second challenge is weak overlap. When the propensity score approaches \(0\) or \(1\), misspecification or estimation error in the propensity model can be strongly amplified through inverse-propensity weighting, obstructing optimal rates and making such procedures unstable. This has motivated a growing line of work that seeks to reduce or bypass direct reliance on propensity-score estimation in limited-overlap regimes \citep{hainmueller2012entropy,imai2014covariate,kim2026estimating}.

At the same time, the standard DML product bound is not always rate-optimal in smoother regimes. For H\"older classes, higher-order influence-function (HOIF) methods show that one can adapt to the complexity of the CATE itself and attain rates strictly faster than the standard product-of-errors heuristic \citep{kennedy2024minimax}. This fast-rate phenomenon is by now well understood in H\"older spaces, but its analogue in \rev{reproducing kernel Hilbert spaces (RKHSs)} remains largely unexplored.
\begin{revised}
To our knowledge, we give the first RKHS analysis of CATE with rates governed by a lower-complexity contrast class that are minimax-optimal in both sample size and overlap, for both $L^2$ and pointwise loss.
\end{revised}

We also study the weak-overlap regime, indexed by a parameter \(\kappa \in [0,1]\) that quantifies how close the propensity score can be to \(0\) or \(1\). The standard pointwise overlap condition that the propensity score lies in \([\kappa,1-\kappa]\) is a representative example, but our analysis also accommodates more singular cases. Rather than treating \(\kappa\) as a fixed constant, we track its effect explicitly throughout the non-asymptotic analysis and establish optimality with respect to \(\kappa\).

\subsection{Contribution}

We study CATE estimation from observational data when the nuisance outcome functions belong to an RKHS \(\cF\), while the CATE function \(h^\star\) belongs to a strictly simpler class \(\cH\). Our goal is to match the statistical difficulty of regression over \(\mathcal{H}\), while still allowing for weak overlap and without imposing structural assumptions on the propensity score beyond overlap itself.
We consider three representative forms of simplicity of the CATE function \(h^\star\): (i) \(\cH\) is an RKHS subspace with faster spectral decay than \(\cF\); (ii) \(h^\star\) satisfies a source condition relative to \(\cF\); and (iii) \(h^\star\) is represented through a known low-dimensional feature map. We study both population \(L^2\) error and pointwise evaluation. Our central message is that, in RKHS outcome models, well-specified nuisance functions are enough to recover learning rates that adapt to the complexity of \(\cH\).

We achieve this with a unified, imputation-based two-stage \rev{kernel ridge regression (KRR)} framework that applies across these three settings. It attains minimax-optimal rates for both \(L^2\) error and pointwise evaluation, with the rate governed by the complexity of \(\mathcal{H}\) rather than by the ambient nuisance space \(\mathcal{F}\). In the appendix, we additionally establish an oracle inequality for a simple model-selection scheme over candidate contrast spaces and regularization levels when \(\mathcal{H}\) is unknown.


Our guarantees remain valid in weak-overlap regimes, where \(\kappa \) may be small and the propensity score may even take the values \(0\) or \(1\) on subsets of the covariate space. In such settings, procedures that rely on propensity-score estimation can be especially fragile under weak overlap and even mild misspecification \citep{hainmueller2012entropy,imai2014covariate,zhou2020propensity}. By contrast, our analysis avoids propensity-score estimation altogether, requires only well-specified outcome models, and still attains rates that are optimal in both \(n\) and \(\kappa\).

We summarize our contributions below:

\begin{itemize}
    \item \textbf{Theory: } We establish optimal dependence on both \(n\) and \(\kappa\), matching the oracle regression rate in \(\mathcal{H}\) with effective sample size \(n\kappa\). We impose no nuisance learning-rate assumption on \(\cF\) beyond well-specifiedness, including for pointwise evaluation. When \(\cH\) is unknown, we also prove a model-selection oracle inequality for the \(L^2\) error.
    \item \textbf{Methodology: } We provide a unified estimation framework based on imputation and two-stage KRR that covers all three structural models. In the appendix, we complement this with a model-selection step over candidate \(\mathcal{H}\) and regularization levels to adapt to unknown contrast complexity.
\end{itemize}

\subsection{Related Work}

\textbf{RKHS in Causal Inference.}\;
\rev{\citet{wang2023statistical,mou2023kernel} studied efficient ATE estimation within the RKHS framework. Related approaches include kernel balancing for partially conditional effects \citep{wang2022double}, kernel ridge Riesz representers \citep{singh2021riesz}, and continuous-treatment effects \citep{kim2026estimating}.} For CATE, \citet{nie2021quasi,foster2023orthogonal} analyzed KRR-based estimators under conditions requiring consistent propensity score estimation and sufficiently fast nuisance learning. \citet{singh2024kernel} proposed RKHS-based methods for CATE, but it remains unclear whether their rates adapt to the complexity of \(h^\star\). 

\textbf{Weak Overlap, Bypassing Propensity Score Estimation.}
Our contribution is to show that, in a broad RKHS setting, the leading rate can be governed by the contrast complexity alone, without relying on propensity score consistency.
This distinction is especially important under weak overlap, where propensity-based procedures can become unstable and highly sensitive to misspecification, motivating more robust weighting and overlap-diagnostic methods for such regimes \citep{zhou2020propensity,ma2020robust,ma2025testing,hainmueller2012entropy,imai2014covariate}. 

\noindent\textbf{Structural Adaptivity.} \;
\citet{kennedy2023towards,kennedy2024minimax,gao2020minimax} established optimal rates for H\"older classes using \rev{DR} estimators and U-statistics. More recent work has investigated particular forms of low-complexity structure, including sparse linear CATEs \citep{kato2023cate} and CATE functions lying in an RKHS with smaller Hilbert norm than the nuisance functions \citep{kim2025transfer}. Standard DML/meta-learner approaches instead emphasize orthogonalization and nuisance-rate products in semiparametric settings \citep{oprescu2019orthogonal,kunzel2019metalearners}. We provide a more detailed comparison with existing methods in Section~\ref{section: comparison}.

\subsection{Notation}
The constants $c, C, c_1, C_1, \dots$ may vary from line to line. We use the symbol $[n]$ as a shorthand for $\{ 1, 2, \dots, n \}$. For nonnegative sequences $\{ a_n \}_{n=1}^{\infty}$ and $\{ b_n \}_{n=1}^{\infty}$, we write $a_n \lesssim b_n$ or $a_n = \mathcal{O}(b_n)$ if there exists a positive constant $C$ such that $a_n \leq C b_n$ for all $n$. We use $\tilde{\mathcal{O}}$ to denote bounds up to polylogarithmic factors. Additionally, we write $a_n \asymp b_n$ if $a_n \lesssim b_n$ and $b_n \lesssim a_n$. For a bounded linear operator $A$, we use $\| A \|_{\mathrm{op}}$ to denote its operator norm. For any $u$ and $v$ in a Hilbert space $\mathbb{H}$, their inner and outer products are denoted by $\langle u, v \rangle_{\mathbb{H}}$ and $u \otimes v$, respectively. When the context is clear, we denote the inner product as \(u^\top v := \langle u, v \rangle_{\mathbb{H}}\).

\section{Problem Setup}\label{section: setup}

\subsection{Treatment Regime and CATE}
We observe \(n\) \rev{independent and identically distributed (i.i.d.)} observational samples \(\mathcal{D} = \{(x_i, a_i, y_i)\}_{i=1}^n\), where \(x_i \in \mathcal{X}\) denotes the covariates, \(a_i \in \{0,1\}\) is a binary treatment indicator, and \(y_i \in \mathbb{R}\) is the response.
For a generic observation \((x,a,y)\), let \(\mathcal{P}_x\) denote the marginal distribution of \(x\) on the covariate space \(\cX\), which we assume to be a regular domain in \(\RR^d\).
We write \(\pi(x) := \PP(a=1 \mid x)\) for the propensity score.

Let \(\mathcal{F}\) be an RKHS containing the nuisance functions \(f_0^\star, f_1^\star\), with kernel \(K_\cF(\cdot,\cdot)\). Our response model is given by:
\begin{align*}
\text{(Nuisances)} \quad
\mathbb{E}[y \mid x, a=0] = f_0^\star(x),
\quad
\mathbb{E}[y \mid x, a=1] = f_1^\star(x).
\end{align*}
Throughout, we assume that the norms \(\|f_0^\star\|_\mathcal{F}\) and \(\|f_1^\star\|_\mathcal{F}\) are bounded by a universal constant \(M>0\).
Let \(\mathcal{H}\) be another function space. Our estimand of interest is the CATE function, defined as:
\begin{align*}
\text{(CATE)} \quad h^\star(x) := f_1^\star(x) - f_0^\star(x).
\end{align*}
Our focus is the regime in which \(h^\star \in \mathcal{H}\) and the hypothesis class \(\mathcal{H}\) is \emph{less complex} than the ambient nuisance space \(\mathcal{F}\).
In this regime, we seek sharp error bounds that adapt to the complexity of \(\mathcal{H}\) and, ideally, match the minimax rates for a regression problem restricted to \(\mathcal{H}\) itself.

We evaluate performance using two metrics: the (squared) $L^2(\mathcal{P}_x)$-error and the (squared) pointwise evaluation error at a fixed query point $x_0 \in \mathcal{X}$:
\begin{align*}
\mathcal{E}_{L^2}(\hat h) := \mathbb{E}_{x\sim\mathcal{P}_x}\bigl|\hat h(x) - h^\star(x)\bigr|^2, \qquad \mathcal{E}_{x_0}(\hat h) := \bigl|\hat h(x_0) - h^\star(x_0)\bigr|^2.
\end{align*}
These are standard measures in nonparametric estimation \citep{wainwright2019high,buhlmann2011statistics}.
\subsection{Three Models of Structural Simplicity}\label{subsection; three models}
We consider three structural scenarios in which the CATE is simpler than the nuisance components. In each case, \(\mathcal{H}\) denotes a target class that is strictly less complex than \(\mathcal{F}\), but the nature of this simplicity differs.

The first scenario is the most direct one: \(\cH\) is an RKHS subspace of \(\cF\) with faster eigenvalue decay.
\begin{tcolorbox}
\begin{model}[Subspace]\label{model; s1}
\(\mathcal{H}\) is an RKHS with \(\mathcal{H} \subset \mathcal{F}\), and \(\mathcal{H}\) exhibits faster spectral decay than \(\mathcal{F}\).
\end{model}
\end{tcolorbox}
Under Model~\ref{model; s1}, \(\cH\) is an RKHS; we denote its associated kernel by \(K_\cH\).
Examples of Model~\ref{model; s1} include:
\begin{itemize}
    \item \textbf{Parametric vs. Nonparametric:} \(\mathcal{F}\) is an infinite-dimensional RKHS (e.g., a Sobolev class), while \(\mathcal{H}\) is a finite-dimensional linear subspace (e.g., polynomials).
    \item \textbf{Differential Sobolev Smoothness:} \(\mathcal{H}\) corresponds to a smoother function class than \(\mathcal{F}\). For instance, \(\mathcal{F} = H^\beta(\mathcal{X})\) and \(\mathcal{H} = H^\gamma(\mathcal{X})\) with \(\gamma \ge \beta > \frac{d}{2}\).
    \item \textbf{Mixed Smoothness:} We can consider mixed Sobolev spaces \citep{kuhn2015approximation,suzuki2018adaptivity} \(\mathcal{F} = H_{\operatorname{mix}}^\beta(\mathcal{X})\) and \(\mathcal{H} = H_{\operatorname{mix}}^\gamma(\mathcal{X})\) with \(\gamma \ge \beta > 1/2\), thereby relaxing the isotropic requirement \(\beta > d/2\).
\end{itemize}

\textbf{Relaxation via Mixed Smoothness.}\;
Isotropic Sobolev spaces \(H^\beta(\cX)\) require \(\beta > d/2\) to define an RKHS, whereas mixed Sobolev spaces \(H_{\mathrm{mix}}^\beta(\cX)\) relax this to \(\beta > 1/2\), independently of \(d\) \rev{\citep{suzuki2018adaptivity,kuhn2015approximation}}. For example, \(H^1_{\mathrm{mix}}\) controls mixed first-order derivatives such as \(\partial_{x_1}\partial_{x_2}\cdots\partial_{x_d} f\) without requiring higher-order derivatives in each coordinate, such as \(\partial_{x_1}^2 f\). This makes \(H_{\mathrm{mix}}^\beta\) a natural model for high-dimensional CATE functions that are smooth through coordinate-wise interactions while avoiding stringent isotropic smoothness.

The second scenario formalizes simplicity through spectral regularity relative to the nuisance kernel.

\begin{tcolorbox}
\begin{model}[Source Condition]\label{model; s2}
We assume \(h^\star\) satisfies a source condition with parameter \(0< \nu <1\) with respect to the kernel of \(\mathcal{F}\).
\end{model}
\end{tcolorbox}

Formally, let $T_{\mathcal{F}}: L^2(\mathcal{P}_x) \to L^2(\mathcal{P}_x)$ be the integral operator associated with the kernel of $\mathcal{F}$. The source condition is \(h^\star \in \operatorname{Range}(T_{\mathcal{F}}^{\frac{1+\nu}{2}})\).
Source conditions are standard in RKHS theory \citep{singh2024kernel,jun2019kernel,chen2024high}; they quantify smoothness by requiring the spectral coefficients of \(h^\star\) to decay faster than those of a generic element of \(\mathcal{F}\).
Equivalently, \(h^\star\) belongs to a fractional power space $[\mathcal{F}]^\nu$, which is strictly smaller and smoother than \(\mathcal{F}\).
This perspective is also central in recent analyses of overparameterized neural networks via the neural tangent kernel (NTK) \citep{zhang2025optimal,ghorbani2021linearized}.

The final scenario captures simplicity through a known low-dimensional representation.

\begin{tcolorbox}
\begin{model}[Low-Dimensional Structure]\label{model; s3}
We assume there exists a known transformation \(\tilde x_i\) of \(x_i\) with \(\dim(\tilde x_i)<d\), and a function \(\tilde h^\star \in \tilde{\mathcal H}\), where \(\tilde{\mathcal H}\) is an RKHS, such that:
\[
h^\star(x_i)=\tilde h^\star(\tilde x_i).
\]
\end{model}
\end{tcolorbox}

Common examples include variable selection, where \(\tilde x_i\) is a subset of coordinates, and linear projections, where \(\tilde x_i = \mathbf{P} x_i\) for a projection \(\mathbf{P}\).
For example, one may take \(\cF = H^\beta(\cX)\) and \(\cH = H^\beta(\tilde{\cX})\), where \(\tilde{\cX}\) is the space of \(\tilde{x}_i\).
Because \(\dim(\tilde{x}_i) < \dim(x_i)\), the statistical complexity of estimating \(\tilde h^\star\) is strictly smaller, and the goal is to adapt to this lower intrinsic dimension.
Under Model~\ref{model; s3}, \(\tilde{\cH}\) is an RKHS; we denote its associated kernel by \(K_{\tilde \cH}\). We set \(\dim(\tilde x_i)=\tilde{d}<d\).

\subsection{Assumptions}\label{section: assumptions}
We impose the following standard assumptions common to causal inference and nonparametric regression \citep{kunzel2019metalearners,kennedy2023towards,wainwright2019high}.

\begin{assumption}[Consistency and Unconfoundedness]\label{ass:id}
Let \(y^0\) and \(y^1\) be the potential outcomes. We observe \(y = y^a\) (consistency), and \((y^0,y^1) \perp\!\!\!\perp a \mid x\) (unconfoundedness).
\end{assumption}

\begin{assumption}[Sub-Gaussian Noise]\label{ass:noise}
\rev{Conditioned on \(x_i,a_i\), the noise is centered and satisfies
\(\EE[\exp(t\varepsilon_i)\mid x_i,a_i]\leq\exp(\sigma^2t^2/2)\) for every \(t\in\mathbb R\), with a uniform constant upper bound \(\sigma\). This permits heteroscedasticity: the conditional variance need not equal \(\sigma^2\) or be constant.} 
\end{assumption}

\begin{assumption}[Boundedness of Kernels]\label{assumption; boundedness}
There exists a universal constant \(\xi>0\) such that
\[
\sup_{x\in\mathcal{X}} K_{\mathcal{F}}(x,x) \le \xi.
\]
\end{assumption}

\begin{assumption}[Weak Overlap]\label{assumption; positivity}
Let \(\phi(\cdot)\) denote the feature map associated with the kernel \(K_\cF\). There exists \(\kappa \in (0,1]\) such that
\begin{align*}
    \EE[\phi(x_i)\otimes \phi(x_i) \bm{1}(a_i=1)] &\preceq \frac{1}{\kappa}\left(\EE[\phi(x_i)\otimes \phi(x_i) \bm{1}(a_i=0)] + \frac{1}{n} \Ib\right), \\
    \EE[\phi(x_i)\otimes \phi(x_i) \bm{1}(a_i=0)] &\preceq \frac{1}{\kappa}\left(\EE[\phi(x_i)\otimes \phi(x_i) \bm{1}(a_i=1)] + \frac{1}{n} \Ib\right).
\end{align*}
\end{assumption}
The parameter \(\kappa\) quantifies the degree of overlap. We do not treat \(\kappa\) as a fixed constant and instead track its effect explicitly in our non-asymptotic bounds. In this sense, our analysis covers the \emph{weak-overlap} regime. Moreover, this condition is automatically satisfied under the usual pointwise overlap assumption \(\kappa < \pi(\cdot) < 1-\kappa\), but it is strictly weaker: it also permits more singular regimes in which \(\pi(x)\) may take the values \(0\) or \(1\) on subsets of the covariate space \citep{wang2026pseudo,kim2025transfer}.

\begin{revised}
For bounded kernels, the $n^{-1}I$ regularization ensures an admissible $\kappa_n\gtrsim n^{-1}$.
The operator condition can remain informative without uniform pointwise overlap: for $X\sim\operatorname{Unif}[0,1]$, $\pi(x)=x$, and $\mathcal F=H^1([0,1])$, one has $\kappa_n\asymp n^{-1/3}$ and hence $n\kappa_n\asymp n^{2/3}\to\infty$ \citep{wang2026pseudo}. Empirical Gram matrices provide a diagnostic of overlap in the chosen kernel space, although they do not certify the population condition.
\end{revised}

\subsection{Fundamental Limits of CATE Estimation}

To place our upper bounds in context, we first record the natural minimax benchmark for CATE estimation. The key comparison is with ordinary nonparametric regression over \(\mathcal{H}\), which provides the appropriate oracle baseline once overlap is taken into account.

\begin{definition}\label{definition:lower-bound}
Let $\texttt{LB-L2}(N; \mathcal{H})$ denote the minimax squared $L^2$-error lower bound for regression over $\mathcal{H}$ with $N$ samples.
Similarly, let $\texttt{LB-PE}(N; \mathcal{H})$ denote the corresponding lower bound for squared pointwise evaluation error.
\end{definition}

For instance, when $\mathcal{H} = H^\gamma(\mathcal{X})$ is a Sobolev space, these lower bounds satisfy $\texttt{LB-L2}(N; \mathcal{H}) \asymp N^{-2\gamma/(d+2\gamma)}$ and $\texttt{LB-PE}(N; \mathcal{H}) \asymp N^{-(2\gamma-d)/2\gamma}$ \citep{tuo2024asymptotic,wainwright2019high}.
With this notation in place, the following informal lemma states the benchmark relevant for CATE estimation.

\begin{lemma}[Informal: Lower Bounds]\label{lemma: lower bound}
The minimax squared $L^2$-error for estimating $h^\star$ is lower-bounded by $\texttt{LB-L2}(n\kappa; \mathcal{H})$, up to constant factors.
Similarly, the minimax squared pointwise evaluation error at $x_0$ is lower-bounded by $\texttt{LB-PE}(n\kappa; \mathcal{H})$, up to constant factors.
\end{lemma}

\begin{proof}
Consider an oracle setting in which the baseline function $f_0^\star$ is known exactly and the propensity score is constant, \(\pi(x) \equiv \kappa\). In this case, estimating \(h^\star\) is equivalent to estimating \(f_1^\star\) using only the treated subsample, whose expected size is \(n\kappa\). The CATE problem therefore reduces to standard regression over \(\mathcal{H}\) with \(n\kappa\) effective observations. Any rate faster than the corresponding regression lower bound would contradict minimax optimality for regression in \(\mathcal{H}\).
\end{proof}

Consequently, for Sobolev spaces \(\cH = H^\gamma(\cX)\), the minimax lower bound for CATE estimation in \(L^2\) error is \((n \kappa)^{-\frac{2\gamma}{d+2\gamma}}\), and the minimax lower bound for pointwise CATE evaluation is \((n\kappa)^{-\frac{2\gamma-d}{2\gamma}}\).

\section{Methodology: A Unified Approach}

This section presents a unified procedure for Models~\labelcref{model; s1,model; s2,model; s3}. The guiding idea is to decouple nuisance estimation from target estimation so that, after a first-stage correction, the second stage behaves like an ordinary regression problem for the contrast \(h^\star\). This is what ultimately allows the final rate to adapt to the complexity of the contrast class.

\subsection{Algorithm Structure}
The procedure has two main ingredients: undersmoothed nuisance estimation and a switch-imputation regression step. The overall algorithm is summarized in Algorithm~\ref{algorithm; main}.

\textbf{1. Nuisance Estimation via Undersmoothed KRR.}\;
We first estimate the conditional mean functions \(f_0^\star\) and \(f_1^\star\) from the observational data \(\mathcal{D}\). The nuisance stage is intentionally undersmoothed so that its regularization bias is small enough not to dominate the final CATE error:
\begin{equation}\label{eq:nuisance-estimation}
\begin{aligned}
\hat f_a &:= \mathop{\arg\min}_{f \in \mathcal{F}} \Bigl\{ \frac{1}{n}\sum_{i=1}^{n} (y_i - f(x_i))^2 \mathbf{1}(a_i=a) + \bar{\lambda}\|f\|_{\mathcal{F}}^2 \Bigr\}, \quad a \in \{0,1\}.
\end{aligned}
\end{equation}
Here, \(\bar{\lambda}\) is the nuisance regularization parameter. It must be small enough to suppress regularization bias, even at the cost of a larger nuisance variance; the key point is that this variance inflation will be absorbed in the second-stage analysis.
For generic bounded kernels, the scaling \(\bar{\lambda} \asymp \log(n)/n\) is sufficient.

\textbf{2. Generating Pseudo-outcomes via Switch-Imputation.} \;
We next form pseudo-outcomes \(\{m_i\}_{i=1}^n\) by imputing the unobserved arm with the first-stage estimator:
\begin{align} \label{equation: switch imputation}
m_{i} :=
\begin{cases}
y_{i} - \hat{f}_0(x_{i}), & \text{if } a_{i} = 1,\\[2pt]
\hat{f}_1(x_{i}) - y_{i}, & \text{if } a_{i} = 0.
\end{cases}
\end{align}
This switch-imputation construction isolates the treatment effect by subtracting the estimated baseline corresponding to the unobserved arm.

\textbf{3. Regression Oracle for the Final Estimator.}\;
Finally, we feed the pseudo-pairs \(\{(x_i,m_i)\}_{i=1}^n\) into a regression oracle \(\mathscr{O}\), which returns
\[
\hat{h} = \mathscr{O}(\{(x_i,m_i)\}_{i=1}^n).
\]
The only model-dependent part of the procedure is the choice of this regression oracle, which reflects the structural assumptions imposed on \(h^\star\); see Section~\ref{sec:oracle-instantiation}.

\begin{algorithm}[]
\caption{Optimal CATE Learner with KRR} \label{algorithm; main}
\begin{algorithmic}[1]
\STATE \textbf{Input:} Dataset \(\mathcal{D}=\{(x_i,a_i,y_i)\}_{i=1}^n\), regression oracle \(\mathscr{O}\), nuisance regularizer \(\bar{\lambda} \asymp \frac{\log n}{n}\), main regularizer $\lambda$.
\STATE \textbf{Step 1: Nuisance Estimation}
\STATE \quad Compute \(\hat f_0, \hat f_1\) via \cref{eq:nuisance-estimation} using \(\mathcal{D}\) and nuisance regularizer \(\bar{\lambda}\).
\STATE \textbf{Step 2: Pseudo-outcome Generation}
\STATE \quad For \(i=1,\dots,n\), compute $m_i$ via \cref{equation: switch imputation} using \(\hat f_0, \hat f_1\).
\STATE \textbf{Step 3: Target Estimation}
\STATE \quad \textbf{Return} \(\hat{h} = \mathscr{O}(\{(x_i,m_i)\}_{i=1}^n)\) with main regularizer \(\lambda\).
\end{algorithmic}
\end{algorithm}

\subsection{Regression Oracle \texorpdfstring{\(\mathscr{O}\)}{O}} \label{sec:oracle-instantiation}
We now instantiate \(\mathscr{O}\) for the three structural models from Section~\ref{section: setup}. In every case, the oracle is a form of KRR applied to the pseudo-outcomes, with the second-stage function space chosen to reflect the complexity of \(h^\star\).

\textbf{Oracle for Model~\ref{model; s1} (Subspace).}\;
Under Model~\ref{model; s1}, where \(h^\star\) resides in a strictly simpler RKHS \(\mathcal{H} \subset \mathcal{F}\), the oracle performs KRR directly within \(\mathcal{H}\):
\[
\hat h := \mathop{\arg\min}_{h \in \mathcal{H}}
\Bigl\{ \frac{1}{n}\sum_{i=1}^n (m_i - h(x_i))^2 + \lambda\|h\|_{\mathcal{H}}^2 \Bigr\}.
\]
The regularization parameter \(\lambda\) is tuned to the spectral complexity of \(\mathcal{H}\), not that of the ambient nuisance space \(\mathcal{F}\).

\textbf{Oracle for Model~\ref{model; s2} (Source Condition).} \;
Under Model~\ref{model; s2}, where \(h^\star \in \mathcal{F}\) satisfies a source condition, the oracle performs KRR in the ambient space \(\mathcal{F}\):
\[
\hat h := \mathop{\arg\min}_{h \in \mathcal{F}}
\Bigl\{ \frac{1}{n}\sum_{i=1}^n (m_i - h(x_i))^2 + \lambda\|h\|_{\mathcal{F}}^2 \Bigr\}.
\]
Although the optimization is carried out in \(\mathcal{F}\), the source condition permits a faster choice of \(\lambda\) and hence faster rates.

\textbf{Oracle for Model~\ref{model; s3} (Low-Dimensional Structure).} \;
Under Model~\ref{model; s3}, where \(h^\star\) depends only on a low-dimensional feature projection \(\tilde x_i\), the oracle performs KRR in the corresponding space \(\tilde{\mathcal{H}}\):
\[
\hat h := \mathop{\arg\min}_{\tilde h \in \tilde{\mathcal{H}}}
\Bigl\{ \frac{1}{n}\sum_{i=1}^n (m_i - \tilde h(\tilde x_i))^2 + \lambda\|\tilde h\|_{\tilde{\mathcal{H}}}^2 \Bigr\}.
\]
This reduces the statistical problem to the intrinsic dimension of the low-dimensional representation.

In all three cases, we refer to \(\lambda\) as the \textbf{main regularizer}.

\begin{revised}
\subsection{Adaptation to Unknown Contrast Structure}\label{sec:adaptation-overview}
The contrast class and its regularization need not be specified in advance. Let $\Pi$ be a finite library of candidate pairs $(\mathcal H,\lambda)$. For example, the library may combine polynomial models of different degrees, Sobolev kernels or NTKs, and models using selected coordinates. Running Algorithm~\ref{algorithm; main} on the first data split for each pair gives a collection $\mathcal M=\{\hat h_1,\ldots,\hat h_L\}$ of candidate CATE estimators, where $L=|\Pi|$.

For the selection guarantee, we truncate each candidate pointwise to $[-B,B]$ by setting $\bar h_j(x)=\min\{B,\max\{-B,\hat h_j(x)\}\}$, and write $\bar{\mathcal M}=\{\bar h_1,\ldots,\bar h_L\}$. Here $B$ is a constant upper bound on $\|h^\star\|_\infty$; exact knowledge of this norm is unnecessary.

Algorithm~\ref{algorithm; model selection} in Appendix~\ref{section; model selection procedure} gives the full model-selection procedure. A second split fits independent nuisance functions, and a third compares the truncated candidates against the resulting switch-imputed responses. We denote the candidate with the smallest empirical squared loss by $\bar h_{\rm ms}\in\bar{\mathcal M}$. This provides a validation criterion even though the true CATE is unobserved.

Under the assumptions of Section~\ref{section: assumptions}, for a library of polynomial size and the regularization specified in Theorem~\ref{theorem: model selection}, the selected estimator satisfies, with probability at least $1-n^{-10}$,
\[
\mathcal E_{L^2}(\bar h_{\rm ms})
\lesssim \min_{\bar h\in\bar{\mathcal M}}\mathcal E_{L^2}(\bar h)
+\widetilde{\mathcal O}\!\left(\frac{1}{n\kappa}\right),
\]
Thus, if the library contains a suitable contrast model and regularization level, selection recovers its rate up to the displayed remainder. This guarantee applies to $L^2$ loss.
\end{revised}

\section{Upper Bounds for \texorpdfstring{$L^2$}{L2}-Error}
In this section, we present minimax-optimal upper bounds on the \(L^2(\mathcal{P}_x)\)-error of the estimator \(\hat{h}\) produced by our meta-algorithm (Algorithm~\ref{algorithm; main}). 
Appendix~\ref{section: general theory} proves Theorem~\ref{theorem: general}, a single general error bound for second-stage KRR in a generic Hilbert space. All of the main corollaries in the paper are then obtained by instantiating this theorem with the appropriate choice of inputs and the evaluation metric.

\subsection{Results under Model~\ref{model; s1} (Subspace)}
We define the kernel integral operator \(\mathbf{T}_{\mathcal{H}}: L^2(\mathcal{P}_x) \to L^2(\mathcal{P}_x)\) associated with the kernel \(K_\cH\) and the marginal distribution \(\mathcal{P}_x\) \citep{fischer2020sobolev,zhang2023optimality}.
Concretely,
\begin{align*}
    \mathbf{T}_\cH f(x) = \int_{x' \in \cX} K(x,x')f(x')\mathrm{d}\cP_x (x').
\end{align*}
Let \(\rho_{\cH,1} \ge \rho_{\cH,2} \ge \cdots > 0\) denote the eigenvalues of \(\mathbf{T}_{\mathcal{H}}\).
The decay rate of these eigenvalues characterizes the complexity of the hypothesis space \(\mathcal{H}\). For instance, for the Sobolev space \(H^\beta(\mathcal{X})\), it is well established that \(\rho_{\cH,j} \asymp j^{-2\beta/d}\) \citep{zhang2023optimality}.
We also assume the kernel is bounded, i.e., \(\sup_{x \in \cX} K_\cH(x,x) = \cO(1)\).

\begin{corollary}[$L^2$-Error Bounds for Model~\ref{model; s1}]\label{corollary; KRR}
Suppose the assumptions in Section~\ref{section: assumptions} and Model~\ref{model; s1} hold.
Then, with probability at least \(1-n^{-10}\), the estimator \(\hat{h}\) from Algorithm~\ref{algorithm; main} with oracle $\Ocr$ for Model~\ref{model; s1} and main regularizer \(\lambda\) satisfies the following bounds:
   \vspace{-0.3cm}
\begin{enumerate}
    \item \textbf{Polynomial Decay:} If \(\rho_{\cH, j} \lesssim j^{-2\ell_\cH}\) for some \(\ell_\cH > 1/2\), setting \(\lambda \asymp (n\kappa)^{-\frac{2\ell_\cH}{2\ell_\cH + 1}}\) yields
    \[
    \mathcal{E}_{L^2}(\hat h) \lesssim (n\kappa)^{-\frac{2\ell_\cH}{2\ell_\cH + 1}}.
    \]
    \item \textbf{Exponential Decay:} If \(\rho_{\cH, j} \lesssim \exp(-c j)\) for some \(c>0\), setting \(\lambda \asymp (n\kappa)^{-1}\) yields 
    \[
    \mathcal{E}_{L^2}(\hat h) \lesssim \frac{1}{n\kappa}.
    \]
    \item \textbf{Finite Rank:} If \(\operatorname{rank}(\mathbf{T}_{\mathcal{H}}) \le D\), setting \(\lambda \asymp (n\kappa)^{-1}\) yields 
    \[
    \mathcal{E}_{L^2}(\hat h) \lesssim \frac{D}{n\kappa}.
    \]
\end{enumerate}
The notation \(\lesssim\) omits absolute constants and polylogarithmic factors.
\end{corollary}
\textbf{Discussion.}\;
We defer the proof to Appendix~\ref{section: apdx proof L2M1}.
Corollary~\ref{corollary; KRR} shows that the estimator adapts to the spectral complexity of the contrast space \(\mathcal{H}\), thereby decoupling the rate from the potentially slower complexity of the nuisance space \(\mathcal{F}\).
Comparing these bounds with Lemma~\ref{lemma: lower bound}, we see that the estimator is minimax optimal in both the sample size \(n\) and the overlap parameter \(\kappa\).
The corresponding Sobolev special cases follow immediately: when $\cF = H^\beta(\cX)$ and \(\cH =H^\gamma(\cX)\), then our bound adapts to $n^{-\frac{2\gamma}{d+2\gamma}}$, which is much faster than the nuisance learning rate \(n^{-\frac{2\beta}{d+2\beta}}\). Notably, on the nuisance side we assume only well-specification of \(\cF\); the eigendecay and rate-determining assumptions are imposed only on the target class \(\cH\).

\subsection{Results under Model~\ref{model; s2} (Source Condition)}
We define the integral operator \(\mathbf{T}_{\mathcal{F}}: L^2(\mathcal{P}_x) \to L^2(\mathcal{P}_x)\) associated with the kernel \(K_\cF\) and the marginal distribution \(\mathcal{P}_x\). Let \(\rho_{\cF,1} \ge \rho_{\cF,2} \ge \cdots > 0\) denote the eigenvalues of \(\mathbf{T}_{\mathcal{F}}\).

\begin{corollary}[$L^2$-Error Bounds for Model~\ref{model; s2}]\label{corollary; source}
Suppose the assumptions in Section~\ref{section: assumptions} hold and \(\rho_{\cF,j} \lesssim j^{-2\ell_\cF}\) for some \(\ell_\cF > 1/2\).
Under Model~\ref{model; s2}, running Algorithm~\ref{algorithm; main} with oracle $\Ocr$ for Model~\ref{model; s2} and main regularizer \(\lambda \asymp (n\kappa)^{-\frac{2\ell_\cF}{1+2\ell_\cF(1+\nu)}}\) yields
\[
\mathcal{E}_{L^2}(\hat h)
\;\lesssim\;
(n\kappa)^{-\frac{2\ell_\cF(1+\nu)}{1+2\ell_\cF(1+\nu)}},
\]
with probability at least \(1-n^{-10}\). The notation \(\lesssim\) omits polylogarithmic factors.
\end{corollary}
\textbf{Discussion.}\;
We defer the proof to Appendix~\ref{section: apdx proof L2M2}.
This result applies to general RKHSs, including settings such as NTKs where adaptivity is far from automatic.
\rev{In the NTK setting of \citet{li2024eigenvalue}, on a bounded domain with bounded covariate density, the spectral exponent \(\ell_\cF=(d+1)/(2d)\) yields $\tilde{\cO}((n\kappa)^{-\frac{(d+1)(1+\nu)}{d+(d+1)(1+\nu)}})$.}
This can be substantially faster than the nuisance NTK learning bound \(n^{-\frac{d+1}{2d+1}}\). Moreover, combined with lower bounds under source conditions \citep{zhang2024optimality}, Lemma~\ref{lemma: lower bound} shows that this rate is optimal.
\rev{Fixed NTKs preserve the i.i.d.\ sampling setup; extending the conclusions to wide networks requires a suitable lazy-training approximation \citep{li2024eigenvalue}.}

\subsection{Results under Model~\ref{model; s3} (Low-Dimensional Structure)}
Finally, we consider the case where \(h^\star\) depends on a known low-dimensional representation of the covariates. Let \(\tilde{\mathbf{T}}\) be the integral operator associated with the kernel on the lower-dimensional space \(\tilde{\cH}\), defined with respect to the distribution of \(\tilde{x}_i\), with eigenvalues \(\tilde{\rho}_1 \ge \tilde{\rho}_2 \ge \dots \). We also assume that the kernel \(K_{\tilde{\cH}}\) is uniformly bounded.
 
\begin{corollary}[$L^2$-Error Bounds for Model~\ref{model; s3}]\label{corollary; model3-L2}
Suppose the assumptions in Section~\ref{section: assumptions} and Model~\ref{model; s3} hold, and \(\tilde\rho_{j} \lesssim j^{-2\tilde \ell}\) for some \(\tilde \ell > 1/2\). Running Algorithm~\ref{algorithm; main} with oracle $\Ocr$ for Model~\ref{model; s3} and main regularizer \(\lambda \asymp (n\kappa)^{-\frac{2\tilde\ell}{1+2\tilde \ell}}\) yields
\begin{align*}
\mathcal{E}_{L^2}(\hat h)
\;\lesssim\;
(n\kappa)^{-\frac{2\tilde \ell}{1+2\tilde \ell}}
\end{align*}
with probability at least $1-n^{-10}$. The notation \(\lesssim\) omits polylogarithmic factors.
\end{corollary}
\textbf{Discussion.}\;
We defer the proof to Appendix~\ref{section: apdx proof L2M3}.
For instance, if \(\tilde{\cH} = H^\beta(\tilde{\cX})\) for \(\operatorname{dim}(\tilde{\cX}) = \tilde{d}\), then \(\tilde{\ell} = \beta / \tilde{d}\), and the rate simplifies to \((n\kappa)^{-\frac{2\beta}{\tilde{d}+2\beta}}\).
This dependence on the intrinsic dimension \(\tilde d\), rather than the ambient dimension \(d\), shows that the estimator behaves as if the target were generated in the lower-dimensional space.
The complexity of the nuisance functions \(f_0^\star\) and \(f_1^\star\), which may depend on the full dimension \(d\), does not degrade the convergence rate of \(\hat{h}\).

\subsection{Comparison with Prior Work}\label{section: comparison}
Closest in spirit to our work is the recent study of \citet{kim2025transfer}, which also considers CATE estimation via a two-stage KRR strategy. The adaptivity considered there, however, is different from ours.
\citet{kim2025transfer} study a ``transfer learning'' regime in which the CATE \(h^\star\) lies in the \emph{same} function space as the nuisance functions (\(\cH = \cF\)), but has a substantially smaller RKHS norm (\(\|h^\star\|_\cF \ll \|f_0^\star\|_\cF, \|f_1^\star\|_\cF\)).
\begin{revised}
For fixed overlap and no covariate shift, the leading term in their bound is, up to logarithmic factors,
\[
n^{-\alpha_\cF}\|h^\star\|_\cF^{2(1-\alpha_\cF)},
\qquad \alpha_\cF=\frac{2\ell_\cF}{1+2\ell_\cF},
\]
where \(\ell_\cF\) is the eigendecay exponent of the common kernel. A smaller contrast norm improves this bound, but the exponent of \(n\) remains fixed by \(\cF\).
In contrast, Models~\ref{model; s1}--\ref{model; s3} allow the exponent itself to adapt to the contrast structure. For example, under the source condition in Model~\ref{model; s2}, Corollary~\ref{corollary; source} improves the exponent from \(\alpha_\cF\) to \(2\ell_\cF(1+\nu)/(1+2\ell_\cF(1+\nu))\). Under Models~\ref{model; s1} and~\ref{model; s3}, we do not impose an eigendecay assumption on the nuisance kernel. We also establish minimax pointwise rates in both \(n\) and \(\kappa\) in Section~\ref{section: upper bounds point evaluation}, a guarantee not established by \citet{kim2025transfer}.
\end{revised}

A more detailed comparison with \rev{DR} and DML-style approaches is given in Appendix~\ref{section: dml comparison}. At a high level, our method is particularly advantageous under weak overlap: it avoids propensity-score estimation entirely, and its statistical burden on the nuisance stage is limited to well-specification, while the eigendecay assumptions governing the rate are imposed only on the contrast class \(\cH\).

\section{Upper Bounds for Point Evaluation}\label{section: upper bounds point evaluation}
In this section, we establish minimax-optimal pointwise error bounds for our estimator. As in the \(L^2\)-analysis, the results here are obtained by applying Theorem~\ref{theorem: general} from Appendix~\ref{section: general theory} under the appropriate evaluation operator.

For pointwise analysis, we focus on Sobolev-type classes, where evaluation at a point can be controlled via Sobolev embedding. The same arguments extend naturally to mixed Sobolev spaces. Throughout this section, we additionally assume that the density of \(x_i\) is bounded above and below by positive constants.

\textbf{Point Evaluation under Models~\labelcref{model; s1,model; s2}.} \;
Assume that \(h^\star\) belongs to a Sobolev space \(\cH = H^\gamma(\cX)\) with smoothness \(\gamma > d/2\), while the nuisance functions may lie in a rougher space \(\cF\).
Under this setup, Model~\ref{model; s2} is subsumed by Model~\ref{model; s1}, since a source condition of order \(\nu\) relative to a Sobolev kernel of smoothness \(\beta\) implies membership in \(H^{\beta(1+\nu)}\).
Accordingly, we analyze Models~\ref{model; s1} and~\ref{model; s2} jointly by working under the Sobolev assumption \(h^\star \in H^\gamma(\cX)\) and applying the Model~\ref{model; s1} oracle.

\noindent\textbf{Point Evaluation under Model~\labelcref{model; s3}.}\;
For Model~\ref{model; s3}, we again work with Sobolev classes and, for simplicity, assume \(\cF = H^\beta(\cX)\) and \(\tilde{\cH} = H^\beta(\tilde{\cX})\) for some \(\beta > d/2\). Recall that \(\operatorname{dim}(\tilde{\cX}) = \tilde{d}\).

We now present the upper bounds for point evaluation under each model.

\begin{corollary}[Point Evaluation: Model~\ref{model; s1} $\&$~\ref{model; s2}]\label{corollary:PE}
Suppose \(\cH = H^\gamma(\cX)\) with \(\gamma > d/2\). Under the assumptions of Section~\ref{section: assumptions} and Model~\ref{model; s1}, the estimator \(\hat{h}\) produced by Algorithm~\ref{algorithm; main} with main regularizer \(\lambda \asymp \frac{1}{n\kappa}\) satisfies the following pointwise error bound for any fixed \(x_0 \in \cX\) with probability at least \(1-n^{-10}\):
\begin{align}\label{eq:pe-bound-1}
\mathcal{E}_{x_0}(\hat h) =| \hat{h}(x_0) - h^\star(x_0) |^2 \;\lesssim\; (n \kappa)^{-\frac{2\gamma-d}{2\gamma}}.
\end{align} Here \(\lesssim\) omits absolute constants and polylogarithmic factors.
\end{corollary}
\textbf{Discussion.}\;
The rate in \cref{eq:pe-bound-1} matches the minimax pointwise rate for estimating a function in \(H^\gamma(\cX)\) from \(n\kappa\) effective observations \citep{tuo2024asymptotic}. Importantly, it is independent of the smoothness of the nuisance functions \(f_0^\star\) and \(f_1^\star\), provided they remain in the ambient RKHS \(\cF\).
The same conclusion holds for mixed Sobolev spaces \(\cH = H^\gamma_{\text{mix}}(\cX), \gamma >\frac{1}{2}\), which yield the rate \((n \kappa)^{-\frac{2\gamma-1}{2\gamma}}\). By Lemma~\ref{lemma: lower bound}, this rate is likewise optimal. We defer the proof to Appendix~\ref{section: apdx proof PEM1}.

Next, we address the low-dimensional setting.

\begin{corollary}[Point Evaluation: Model~\ref{model; s3}]\label{corollary:PE-model3}
Under the assumptions of Section~\ref{section: assumptions} and Model~\ref{model; s3}, suppose \(\tilde{\cH} = H^\beta(\tilde{\cX})\) and \(\cF = H^\beta(\cX)\) with \(\beta > d/2\).
Then the estimator \(\hat h\) produced by Algorithm~\ref{algorithm; main} with main regularizer \(\lambda \asymp \frac{1}{n\kappa}\) satisfies:
\begin{align}\label{eq:pe-bound-3}
\mathcal{E}_{x_0}(\hat h) =| \hat{h}(x_0) - h^\star(x_0) |^2 \;\lesssim\; (n \kappa)^{-\frac{2\beta-\tilde d}{2\beta}}
\end{align}
with probability at least \(1-n^{-10}\). Here \(\lesssim\) omits absolute constants and polylogarithmic factors.
\end{corollary}
\textbf{Discussion.}\;
This bound depends on the intrinsic dimension \(\tilde d\) rather than the ambient dimension \(d\), thereby showing that the method adapts to low-dimensional structure in the target even when nuisance learning takes place in the full space.
We defer the proof to Appendix~\ref{section: apdx proof PEM3}.

\textbf{Comparison with DML.}\;
Relative to pointwise \rev{DR} analyses such as \citet{kennedy2023towards}, our pointwise guarantees avoid a substantial burden. Those results typically require sufficiently fast control of pointwise nuisance errors such as \(|\hat{f}(x_0)-f^\star(x_0)|\) and \(|\hat{\pi}(x_0)-\pi(x_0)|\), whereas standard learning theory more naturally provides \(L^2\)-type guarantees. Our approach bypasses propensity estimation altogether and requires only that the outcome nuisances be well specified in \(\mathcal{F}\), while the leading pointwise rate is governed solely by the complexity of the target class \(\mathcal{H}\).

\begin{revised}
\section{Experiments}\label{sec:discussion}
We compare Ours with plug-in KRR, the DR-learner of \citet{kennedy2023towards}, and the R-learner of \citet{nie2021quasi}. The R-learner fits the treatment effect through a residualized regression loss. Synthetic experiments additionally include oracle DR, which uses the true propensity score. We also use RealCause \citep{neal2020realcause}, which fits generative models to real data to produce treatment assignments and potential outcomes given observed covariates. Our semi-real experiments use its LaLonde benchmark, based on the National Supported Work job-training study and a nonexperimental comparison group from the Current Population Survey (CPS) \citep{lalonde1986evaluating}. The outcome is post-training earnings; RealCause uses the data version of \citet{dehejia1999causal}.

Ours remains competitive with DR, oracle DR, and R in the univariate examples, where low-dimensional nuisance estimation is relatively favorable to DR-based methods. Its advantage is clearer in the multivariate experiments: it attains the lowest mean error in all sample-size and noise comparisons, as well as all eight settings with $\kappa\in\{0.01,0.05\}$. Noise, overlap, nuisance-kernel misspecification, and representation ablations examine the sensitivity of these findings. On 50 RealCause realizations, Ours also has the lowest average held-out mean squared error (MSE). Subtracting the common estimated outcome-noise contribution makes the excess-risk differences easier to compare and leaves the ranking unchanged. Appendix~\ref{section; experiments} gives the experimental setup and results in Tables~\ref{table:univariate}--\ref{table:realcause}.
Code is available at \url{https://github.com/seokjinkim0428/Structure-Adaptive-CATE}.

\end{revised}

\begin{revised}
\renewcommand{\acksection}{\section*{Acknowledgments}}
\begin{ack}
We thank Kaizheng Wang for helpful discussions.
We also thank the reviewers for their constructive feedback, which helped improve the paper.
\end{ack}
\end{revised}

\begingroup
\let\originalbibitem\bibitem
\def\revisedbibone{wang2022double}
\def\revisedbibtwo{singh2021riesz}
\def\revisedbibthree{lalonde1986evaluating}
\def\revisedbibfour{dehejia1999causal}
\def\revisedbibfive{neal2020realcause}
\renewcommand{\bibitem}[2][]{%
  \normalcolor\def\currentbibkey{#2}%
  \ifx\currentbibkey\revisedbibone\revisioncolor\fi
  \ifx\currentbibkey\revisedbibtwo\revisioncolor\fi
  \ifx\currentbibkey\revisedbibthree\revisioncolor\fi
  \ifx\currentbibkey\revisedbibfour\revisioncolor\fi
  \ifx\currentbibkey\revisedbibfive\revisioncolor\fi
  \originalbibitem[#1]{#2}}
\bibliography{ref}
\endgroup

\newpage
\appendix

\clearpage
\newpage
\begin{center}
\textbf{{\Huge Supplementary Materials}}
\end{center}
\vspace{0.5cm}
\etocdepthtag.toc{mtappendix}
\etocsettagdepth{mtchapter}{none}
\etocsettagdepth{mtappendix}{subsection}
\appendix
\numberwithin{theorem}{section}
\setcounter{theorem}{0}
\tableofcontents

\crefalias{section}{appendix} 

\paragraph{Appendix Roadmap.}\;
For convenience, we summarize the supporting material as follows. Section~\ref{section; model selection procedure} presents the model-selection procedure and its oracle guarantee. Section~\ref{section; experiments} reports the experimental details and additional numerical results. \rev{Section~\ref{section: general theory} develops the general Hilbert-space error bound, followed by its proofs and technical lemmas. Section~\ref{sec:misspecification} extends the analysis to approximate outcome models and proves the additional error bound. Section~\ref{sec:limitations} discusses limitations.}

\section{Detailed Comparison with DML Methods}\label{section: dml comparison}
Our weak-overlap formulation is broad enough to allow singular propensity scores, thereby covering regimes in which standard \rev{DR} methods are not directly applicable. Even in the classical regime \(\pi(\cdot)\in[\kappa,1-\kappa]\), however, our approach retains substantial advantages over \rev{DR} methods.
Avoiding propensity-score estimation is itself a substantial advantage, because propensity models are well known to be especially sensitive to weak overlap and even mild model misspecification \citep{hainmueller2012entropy,imai2014covariate,zhou2020propensity}.
At a high level, existing DR analyses typically yield squared error bounds of the form
\begin{equation*}
    \kappa^{-A}\bigl((\text{Oracle rate for }\mathcal{H}\text{ with }n\kappa\text{ samples})+\kappa^{-B}\|\hat f_a-f_a^\star\|_{\mathbb{X}}^2\|\hat \pi-\pi\|_{\mathbb{X}}^2\bigr).
\end{equation*}
For instance, \citet[Eq.~(12)]{nie2021quasi} corresponds to \((\mathbb{X},A,B)=(L^2,2,0)\), while \citet[Proposition~2]{foster2023orthogonal} corresponds to \((\mathbb{X},A,B)=(L^4,0,3)\). In both cases, the resulting guarantees are highly vulnerable under weak overlap, that is, when \(\kappa\) is small.
\begin{itemize}
    \item \citet{nie2021quasi}: the leading term is not optimal in \(\kappa\) because of the prefactor \(\kappa^{-A}\), even when the nuisance estimators converge sufficiently fast. By contrast, our rates are optimal in both \(n\) and \(\kappa\). Their analysis also requires \(\|\hat \pi-\pi\|_{L^2}^2=o(1/\sqrt{n})\).
    \item \citet{foster2023orthogonal}: the required \(L^4\) nuisance error bounds are stronger than standard \(L^2\) control, and the additional factor \(\kappa^{-3}\) makes the bound unstable under weak overlap. In contrast, we require only that \(f^\star\) be well specified and do not estimate \(\pi\). Moreover, these methods guarantee only CATE \(L^2\) error, whereas our algorithm also yields pointwise bounds.
    \item \citet{kennedy2023towards}: this line of work focuses on pointwise bounds. In that setting, one may view \(\mathbb{X}=L^2(\delta_{x_0})\) and \((A,B)=(0,2)\). It requires pointwise error control for both \(f^\star\) and \(\pi\), whereas standard learning theory more naturally provides \(L^2\)-type guarantees. Even for point evaluation, we require only that \(f^\star\) be well specified. \rev{We impose no separate nuisance-estimator error bound, and the target spectrum determines the leading rate.}
\end{itemize}

\section{Model Selection Procedure and Adaptivity Guarantees}
\label{section; model selection procedure}
\subsection{Model Selection Procedure}
In this section, we present a model-selection procedure and its corresponding guarantee for \(L^2\) error. In practice, the best contrast space \(\mathcal{H}\) and the corresponding regularization level are rarely known a priori. While the nuisance class \(\mathcal{F}\) can be tuned against observed outcomes, selecting a CATE model is harder because the target itself is not directly observed. We therefore introduce a dedicated model-selection procedure that chooses among a library of candidate estimators.

The procedure uses three disjoint data splits so that candidate generation, proxy construction, and final selection remain independent.

\textbf{1. Candidate Generation (\(\mathcal{D}_1\)).}
We partition the dataset \(\mathcal{D}\) into three disjoint subsets \(\mathcal{D}_1, \mathcal{D}_2, \mathcal{D}_3\). For simplicity, we analyze the equal-split case \(n_1=n_2=n_3=n/3\); the same conclusions hold for any constant-proportion split.
Using \(\mathcal{D}_1\), we run Algorithm~\ref{algorithm; main} over a library of configurations, for example different hypothesis classes \(\{\mathcal{H}_1, \dots, \mathcal{H}_K\}\) and regularization levels, to generate candidate estimators \(\mathcal{M} = \{\hat{h}_1, \dots, \hat{h}_L\}\).
We denote the configuration library by \(\Pi\), where each element is a pair \((\cH,\lambda)\) consisting of a hypothesis class and a regularization level.

\textbf{2. Truncation.}
To control the variance during the selection phase, we enforce a boundedness constraint.
Let \(B>0\) be the truncation level. For each \(\hat{h}_j \in \mathcal{M}\), we define its truncated version \(\bar{h}_j(x) := \min(\max(\hat{h}_j(x), -B), B)\). Let \(\bar{\mathcal{M}} = \{\bar{h}_1, \dots, \bar{h}_L\}\) denote the set of truncated candidates.

\textbf{3. Proxy Construction (\(\mathcal{D}_2\)).}
Using \(\mathcal{D}_2\), we construct a high-quality proxy for the unobserved CATE by re-estimating the nuisance functions with an independent sample and undersmoothed KRR:
\begin{align}\label{equation: tilde f}
\tilde{f}_a &:= \mathop{\arg\min}_{f \in \mathcal{F}} \Bigl\{ \frac{1}{n_2}\sum_{i =1}^{n_2} (y_{2i} - f(x_{2i}))^2 \mathbf{1}(a_{2i}=a) + \tilde{\lambda}\|f\|_{\mathcal{F}}^2 \Bigr\}, \quad a \in \{0,1\}.
\end{align}
These nuisance estimates depend only on \(\mathcal{D}_2\) and are therefore independent of the candidate estimators generated from \(\mathcal{D}_1\).

\textbf{4. Empirical Risk Minimization (\(\mathcal{D}_3\)).}
Finally, using \(\mathcal{D}_3\), we evaluate the candidates against the proxy induced by \(\tilde f_0\) and \(\tilde f_1\). For each \((x_{3i}, a_{3i}, y_{3i}) \in \mathcal{D}_3\), define
\begin{align}\label{equation: tilde m}
\tilde m_{3i} :=
\begin{cases}
y_{3i} - \tilde{f}_0(x_{3i}), & \text{if } a_{3i} = 1,\\[2pt]
\tilde{f}_1(x_{3i}) - y_{3i}, & \text{if } a_{3i} = 0.
\end{cases}
\end{align}
We select the estimator that minimizes the empirical squared error with respect to these proxies:
\[
\rev{\bar{h}_{\operatorname{ms}}} := \mathop{\arg\min}_{\bar{h} \in \bar{\mathcal{M}}} \frac{1}{n_3} \sum_{i =1}^{n_3} (\bar{h}(x_{3i}) - \tilde{m}_{3i})^2.
\]

\begin{algorithm}[h]
\caption{CATE Model Selection}
\label{algorithm; model selection}
\begin{algorithmic}[1]
\STATE \textbf{Input:} Data \(\mathcal{D}\), set of candidate configurations \(\Pi\)\rev{, truncation bound \(B\)}.
\STATE Partition \(\mathcal{D}\) into \(\mathcal{D}_1, \mathcal{D}_2, \mathcal{D}_3\).
\STATE \textbf{Step 1:} Train candidates \(\mathcal{M} = \{\hat{h}_1, \dots, \hat{h}_L\}\) on \(\mathcal{D}_1\) by running Algorithm~\ref{algorithm; main} over multiple configurations \((\cH, \lambda) \in \Pi\).
\STATE \textbf{Step 2:} Truncate candidates to \([-B, B]\) to obtain \(\bar{\mathcal{M}} = \{\bar{h}_1, \dots \bar{h}_L\}\).
\STATE \textbf{Step 3:} Estimate nuisance functions \(\tilde{f}_0, \tilde{f}_1\) on \(\mathcal{D}_2\) via \cref{equation: tilde f}.
\STATE \textbf{Step 4:} For \((x_{3i},a_{3i},y_{3i}) \in \mathcal{D}_3\), compute proxies \(\tilde{m}_{3i}\) using \(\tilde{f}_0,\tilde{f}_1\) by \cref{equation: tilde m}.
\STATE \textbf{Return} \(\bar{h}_{\operatorname{ms}} = \arg\min_{\bar{h} \in \bar{\mathcal{M}}} \frac{1}{n_3} \sum_{i =1}^{n_3} (\bar{h}(x_{3i}) - \tilde{m}_{3i})^2\).
\end{algorithmic}
\end{algorithm}

\subsection{Adaptivity Guarantees}
\rev{We now give the oracle guarantee for Algorithm~\ref{algorithm; model selection}.}

\begin{revised}
\begin{assumption}[Bounded contrast]\label{assumption:bounded-contrast}
There is a finite $B^\star$ such that $\|h^\star\|_\infty\leq B^\star$.
\end{assumption}
Choose $B\geq B^\star$ in advance or using $\mathcal D_1$ or independent data; exact knowledge of $B^\star$ is unnecessary.

\begin{theorem}[Oracle inequality]\label{theorem: model selection}
Suppose Assumptions~\labelcref{ass:id,ass:noise,assumption; boundedness,assumption:bounded-contrast,assumption; positivity} hold, with bounded nuisance RKHS norms. Run Algorithm~\ref{algorithm; model selection} with constant-proportion splits, $\tilde\lambda\asymp\log n/n_2$, and $B\geq B^\star$ as above. If $L=|\mathcal M|$ is polynomial in $n$, then with probability at least $1-n^{-10}$,
\begin{equation}\label{eq:selection-B-cost}
\mathcal E_{L^2}(\bar h_{\rm ms})\lesssim
\min_{\bar h\in\bar{\mathcal M}}\mathcal E_{L^2}(\bar h)
+\widetilde{\mathcal O}\!\left(\frac{1}{n\kappa}+\frac{(B+B^\star)^2}{n}\right).
\end{equation}
\end{theorem}
Bounded upper estimates of $B^\star$ preserve the nonparametric oracle rate whenever the parametric-order remainder is lower order. Appendix~\ref{sec:selection-proof} gives the derivation, including logarithmic factors.

\subsection{Practical choices}\label{sec:practical-guidance}
For coordinate subsets, $\log|\mathcal M|=\log\binom{d}{\tilde d}\le\tilde d\log(ed/\tilde d)$ before adding kernel and regularization grids. For principal component regression, a grid of at most $d$ leading subspaces constructed on $\mathcal D_1$ contributes only $\log d$. The guarantee compares supplied candidates; learning the unknown projection remains open.

For pointwise tuning, choose $\lambda\asymp(n\kappa)^{-1}$; a constant-factor approximation of $\kappa$ suffices. We recommend averaging the cyclic fits.

\end{revised}

\section{Numerical Experiments}\label{section; experiments}
\begin{revised}
We compare Ours with plug-in KRR, DR \citep{kennedy2023towards}, oracle-propensity DR, and R \citep{nie2021quasi}, using the same datasets and test covariates for all methods. HOIF methods serve as theoretical benchmarks.

The raw ridge grid is $\mathcal A=\{1,2,4,\ldots,512\}$ in $(K+\alpha I)^{-1}y$. In our experiments, Ours fixes the raw nuisance ridge at $c=0.01$, corresponding to normalized regularization $c/n_{\rm stage}$, and selects the target kernel and ridge by Algorithm~\ref{algorithm; model selection}. We use untruncated predictions and average the three cyclic fits. Plug-in uses three-fold cross-validation (CV) and refits each arm. DR, oracle DR, and R use three-fold CV within two swapped fits and average their predictions. R minimizes the residualized loss $(Y-\hat m(X)-(A-\hat\pi(X))h(X))^2$, with $m(x)=\mathbb E[Y\mid X=x]$. Estimated propensities are clipped to $[0.01,0.99]$; oracle DR uses the true propensity.

We generate $A\mid X\sim\operatorname{Bernoulli}(\pi_\kappa(X))$ and $Y=f_0^\star(X)+Ah^\star(X)+\varepsilon$, where $\pi_\kappa(x)=\operatorname{Clip}(\sin(5\|x\|_2),\kappa,1-\kappa)$ and $\varepsilon\sim\mathcal N(0,\sigma^2)$ is independent. Defaults are $\sigma=1$ and $\kappa=0.1$. Tables report mean integrated squared error (MISE) with Monte Carlo standard errors (SEs) on 3,000 independent uniform test points (seed $10099$), excluded from tuning. Sample-size comparisons use 100 replications; sweeps and ablations use 30, with seed $100+10r$ for replication $r$.

\subsection{Univariate Synthetic Data}
We draw $X\sim\operatorname{Unif}[0,1]$ and set
\[
f_0^\star(x)=5\bigl(|x-0.4|+|x-0.8|\bigr),\qquad h^\star(x)=x^2.
\]
The default nuisance/target kernels are unanchored Sobolev kernels of orders $(H^1,H^2)$. We vary $n\in\{500,1000\}$ and, at $n=500$, noise $\sigma\in\{0.25,0.5,1,2,4\}$ and overlap $\kappa\in\{0.01,0.02,0.05,0.1,0.2,0.3,0.5\}$. The $(H^2,H^3)$ ablation at $n=1000$ misspecifies the kinked nuisance while retaining a smooth target; it changes both kernels.

This univariate setting is relatively favorable to DR-based methods, since low-dimensional outcome and propensity estimation can make their product-of-errors remainder small. Ours remains comparable to DR, oracle DR, and R (Table~\ref{tab:univariate-all}): at $n=1000$, their MISEs are $0.0424$, $0.0577$, $0.0365$, and $0.0335$, respectively.

\begin{table}[H]
\revisioncolor
\centering
\caption{Univariate synthetic results. Entries are mean MISE (SE), with the smallest mean in bold. Unless varied, $\sigma=1$ and $\kappa=0.1$; the $\sigma=1$ row also supplies the $\kappa=0.1$ reference. $H^a/H^b$ denotes the jointly varied nuisance/target Sobolev orders; $H^2$ misspecifies the kinked nuisance.}
\label{tab:univariate-all}
\label{table:univariate}
\fontsize{8.5}{10}\selectfont
\setlength{\tabcolsep}{2pt}
\begin{tabular}{@{}lrrccccc@{}}
\toprule
Configuration & $n$ & Reps. & Ours & Plug-in & DR & Oracle DR & R \\
\midrule
Sample size & 500 & 100 & 0.0811\,(0.0056) & 0.1581\,(0.0087) & 0.3203\,(0.1599) & 0.0758\,(0.0057) & \textbf{0.0699}\,(0.0051) \\
Sample size & 1000 & 100 & 0.0424\,(0.0033) & 0.0825\,(0.0035) & 0.0577\,(0.0204) & 0.0365\,(0.0033) & \textbf{0.0335}\,(0.0031) \\
$\sigma=0.25$ & 500 & 30 & 0.0317\,(0.0061) & 0.1230\,(0.0104) & 0.0938\,(0.0631) & 0.0527\,(0.0076) & \textbf{0.0210}\,(0.0025) \\
$\sigma=0.5$ & 500 & 30 & 0.0473\,(0.0078) & 0.1388\,(0.0127) & 0.0828\,(0.0398) & 0.0688\,(0.0100) & \textbf{0.0368}\,(0.0055) \\
$\sigma=1$ & 500 & 30 & 0.0890\,(0.0098) & 0.1908\,(0.0201) & \textbf{0.0704}\,(0.0119) & 0.0983\,(0.0132) & 0.0814\,(0.0106) \\
$\sigma=2$ & 500 & 30 & \textbf{0.1406}\,(0.0141) & 0.3814\,(0.0451) & 0.1585\,(0.0237) & 0.1839\,(0.0227) & 0.1697\,(0.0162) \\
$\sigma=4$ & 500 & 30 & \textbf{0.2701}\,(0.0454) & 1.0014\,(0.0945) & 0.5051\,(0.1317) & 0.4994\,(0.1034) & 0.3686\,(0.0707) \\
$\kappa=0.01$ & 500 & 30 & 0.2939\,(0.0382) & 0.6344\,(0.0584) & 0.3802\,(0.0961) & 0.3739\,(0.1063) & \textbf{0.1601}\,(0.0143) \\
$\kappa=0.02$ & 500 & 30 & 0.2263\,(0.0319) & 0.4839\,(0.0480) & 0.2363\,(0.0416) & 0.2434\,(0.0423) & \textbf{0.1381}\,(0.0122) \\
$\kappa=0.05$ & 500 & 30 & 0.1324\,(0.0183) & 0.2983\,(0.0367) & 0.2197\,(0.0574) & 0.1437\,(0.0176) & \textbf{0.1058}\,(0.0126) \\
$\kappa=0.2$ & 500 & 30 & 0.0634\,(0.0103) & 0.0970\,(0.0105) & 0.0678\,(0.0109) & 0.0676\,(0.0095) & \textbf{0.0539}\,(0.0068) \\
$\kappa=0.3$ & 500 & 30 & 0.0495\,(0.0091) & 0.0671\,(0.0065) & 0.0469\,(0.0072) & 0.0548\,(0.0085) & \textbf{0.0414}\,(0.0067) \\
$\kappa=0.5$ & 500 & 30 & 0.0440\,(0.0065) & 0.0415\,(0.0039) & \textbf{0.0347}\,(0.0047) & 0.0370\,(0.0051) & 0.0379\,(0.0056) \\
$H^1/H^2$ & 1000 & 30 & 0.0459\,(0.0056) & 0.0790\,(0.0064) & 0.0440\,(0.0086) & \textbf{0.0297}\,(0.0051) & 0.0364\,(0.0073) \\
$H^2/H^3$ & 1000 & 30 & \textbf{0.0493}\,(0.0067) & 0.3063\,(0.0249) & 6.0731\,(1.6719) & 0.0499\,(0.0076) & 0.0850\,(0.0049) \\
\bottomrule
\end{tabular}
\end{table}
\subsection{Multivariate Synthetic Data}
We draw $X\sim\operatorname{Unif}([-1,1]^d)$ and set
\[
f_0^\star(x)=\frac{2}{d}\sum_{j=1}^{d}\sin(x_j),\qquad
h^\star(x)=\frac{0.3}{p}\sum_{j=1}^{p}x_j^2.
\]
We consider dense ($p=d$) and sparse ($p=4$) targets with $d=10,20$. The nuisance uses a Mat\'ern kernel with $\nu=1.5$. Ours selects among Mat\'ern ($\nu=1.5,2.5$) and radial basis function kernels on all coordinates and on $x_{1:4}$. DR, oracle DR, and R use Mat\'ern target kernels with $\nu=2.5$ and receive the true support in sparse settings. Sample-size comparisons use fixed kernel lengths; sweeps use lengths calibrated to correlation $0.5$ at the median distance on 2,000 independent covariates (seed $42$). Sweeps use $n=1000$, $\sigma\in\{0.5,1,2\}$, and $\kappa\in\{0.01,0.05,0.1,0.2,0.5\}$.

The advantage of Ours is clearer in Table~\ref{tab:multivariate-all}: it has the lowest mean in 4 of the 4 sample-size comparisons and all 12 noise conditions. At $\sigma=1$, its error is approximately $77$--$83\%$ below plug-in. R or DR leads in some moderate-overlap settings. Oracle DR need not dominate estimated-propensity DR in finite samples, since exact propensities can increase weighting variance.

For $(d,p)=(20,4)$, adding irrelevant $x_{5:8}$ candidates barely changes Ours ($0.0092$ to $0.0094$). Using only these incorrect coordinates gives $0.0065$, a finite-sample result rather than evidence of support recovery. Baselines keep their true-support kernels.

\begin{table}[htbp]
\revisioncolor
\centering
\caption{Multivariate synthetic results. Entries are mean MISE (SE), with the smallest mean in bold. Unless varied, $\sigma=1$ and $\kappa=0.1$; the $\sigma=1$ row also supplies the $\kappa=0.1$ reference. $p=d$ is dense and $p=4$ is sparse. Sample-size rows use fixed kernel lengths; other rows use median-calibrated lengths. The $\sigma=1$ row at $(d,p)=(20,4)$ also gives the default representation library.}
\label{tab:multivariate-all}
\label{table; multivariate n=1000}
\label{table; multivariate n=2000}
\fontsize{8.2}{10}\selectfont
\setlength{\tabcolsep}{1.8pt}
\begin{tabular}{@{}llrrccccc@{}}
\toprule
$(d,p)$ & Configuration & $n$ & Reps. & Ours & Plug-in & DR & Oracle DR & R \\
\midrule
$(10,10)$ & Sample size & 1000 & 100 & \textbf{0.0099}\,(0.0008) & 0.0478\,(0.0018) & 0.0184\,(0.0049) & 0.0201\,(0.0031) & 0.0127\,(0.0012) \\
$(10,10)$ & Sample size & 2000 & 100 & \textbf{0.0062}\,(0.0004) & 0.0295\,(0.0010) & 0.0175\,(0.0041) & 0.0099\,(0.0012) & 0.0076\,(0.0006) \\
$(10,10)$ & $\sigma=0.5$ & 1000 & 30 & \textbf{0.0051}\,(0.0005) & 0.0204\,(0.0011) & 0.0054\,(0.0008) & 0.0087\,(0.0022) & 0.0058\,(0.0007) \\
$(10,10)$ & $\sigma=1$ & 1000 & 30 & \textbf{0.0081}\,(0.0011) & 0.0479\,(0.0032) & 0.0119\,(0.0030) & 0.0222\,(0.0072) & 0.0134\,(0.0023) \\
$(10,10)$ & $\sigma=2$ & 1000 & 30 & \textbf{0.0170}\,(0.0033) & 0.1047\,(0.0124) & 0.0377\,(0.0122) & 0.0570\,(0.0207) & 0.0382\,(0.0090) \\
$(10,10)$ & $\kappa=0.01$ & 1000 & 30 & \textbf{0.0088}\,(0.0010) & 0.0511\,(0.0040) & 0.0102\,(0.0015) & 0.1019\,(0.0790) & 0.0128\,(0.0018) \\
$(10,10)$ & $\kappa=0.05$ & 1000 & 30 & \textbf{0.0078}\,(0.0009) & 0.0517\,(0.0045) & 0.0113\,(0.0020) & 0.0213\,(0.0068) & 0.0124\,(0.0020) \\
$(10,10)$ & $\kappa=0.2$ & 1000 & 30 & 0.0107\,(0.0016) & 0.0462\,(0.0039) & 0.0104\,(0.0018) & 0.0102\,(0.0013) & \textbf{0.0096}\,(0.0011) \\
$(10,10)$ & $\kappa=0.5$ & 1000 & 30 & \textbf{0.0069}\,(0.0011) & 0.0450\,(0.0036) & 0.0127\,(0.0034) & 0.0121\,(0.0032) & 0.0089\,(0.0014) \\
\addlinespace[3pt]
$(10,4)$ & Sample size & 1000 & 100 & \textbf{0.0112}\,(0.0008) & 0.0491\,(0.0018) & 0.0171\,(0.0024) & 0.0186\,(0.0021) & 0.0141\,(0.0011) \\
$(10,4)$ & Sample size & 2000 & 100 & \textbf{0.0074}\,(0.0003) & 0.0307\,(0.0010) & 0.0121\,(0.0014) & 0.0126\,(0.0010) & 0.0088\,(0.0005) \\
$(10,4)$ & $\sigma=0.5$ & 1000 & 30 & \textbf{0.0062}\,(0.0005) & 0.0212\,(0.0011) & 0.0066\,(0.0005) & 0.0085\,(0.0016) & 0.0079\,(0.0008) \\
$(10,4)$ & $\sigma=1$ & 1000 & 30 & \textbf{0.0092}\,(0.0010) & 0.0489\,(0.0032) & 0.0117\,(0.0018) & 0.0174\,(0.0052) & 0.0171\,(0.0029) \\
$(10,4)$ & $\sigma=2$ & 1000 & 30 & \textbf{0.0183}\,(0.0033) & 0.1061\,(0.0124) & 0.0306\,(0.0067) & 0.0542\,(0.0220) & 0.0476\,(0.0103) \\
$(10,4)$ & $\kappa=0.01$ & 1000 & 30 & \textbf{0.0098}\,(0.0009) & 0.0523\,(0.0040) & 0.0180\,(0.0041) & 0.0496\,(0.0156) & 0.0167\,(0.0026) \\
$(10,4)$ & $\kappa=0.05$ & 1000 & 30 & \textbf{0.0090}\,(0.0009) & 0.0527\,(0.0045) & 0.0131\,(0.0028) & 0.0261\,(0.0096) & 0.0156\,(0.0025) \\
$(10,4)$ & $\kappa=0.2$ & 1000 & 30 & 0.0112\,(0.0014) & 0.0469\,(0.0039) & 0.0099\,(0.0011) & 0.0119\,(0.0017) & \textbf{0.0098}\,(0.0009) \\
$(10,4)$ & $\kappa=0.5$ & 1000 & 30 & \textbf{0.0083}\,(0.0011) & 0.0459\,(0.0036) & 0.0106\,(0.0017) & 0.0109\,(0.0016) & 0.0100\,(0.0014) \\
\addlinespace[3pt]
$(20,20)$ & $\sigma=0.5$ & 1000 & 30 & \textbf{0.0032}\,(0.0003) & 0.0204\,(0.0011) & 0.0039\,(0.0005) & 0.0051\,(0.0011) & 0.0035\,(0.0004) \\
$(20,20)$ & $\sigma=1$ & 1000 & 30 & \textbf{0.0074}\,(0.0009) & 0.0388\,(0.0031) & 0.0093\,(0.0016) & 0.0125\,(0.0038) & 0.0086\,(0.0009) \\
$(20,20)$ & $\sigma=2$ & 1000 & 30 & \textbf{0.0212}\,(0.0035) & 0.0602\,(0.0099) & 0.0289\,(0.0062) & 0.0367\,(0.0115) & 0.0237\,(0.0036) \\
$(20,20)$ & $\kappa=0.01$ & 1000 & 30 & \textbf{0.0072}\,(0.0009) & 0.0417\,(0.0036) & 0.0100\,(0.0017) & 0.0311\,(0.0101) & 0.0097\,(0.0015) \\
$(20,20)$ & $\kappa=0.05$ & 1000 & 30 & \textbf{0.0076}\,(0.0012) & 0.0401\,(0.0038) & 0.0095\,(0.0016) & 0.0173\,(0.0034) & 0.0099\,(0.0013) \\
$(20,20)$ & $\kappa=0.2$ & 1000 & 30 & \textbf{0.0077}\,(0.0013) & 0.0388\,(0.0034) & 0.0114\,(0.0024) & 0.0104\,(0.0025) & 0.0092\,(0.0014) \\
$(20,20)$ & $\kappa=0.5$ & 1000 & 30 & 0.0089\,(0.0012) & 0.0408\,(0.0028) & \textbf{0.0075}\,(0.0010) & 0.0076\,(0.0010) & 0.0104\,(0.0019) \\
\addlinespace[3pt]
$(20,4)$ & $\sigma=0.5$ & 1000 & 30 & \textbf{0.0050}\,(0.0004) & 0.0221\,(0.0011) & 0.0061\,(0.0004) & 0.0071\,(0.0009) & 0.0061\,(0.0004) \\
$(20,4)$ & $\sigma=1$ & 1000 & 30 & \textbf{0.0092}\,(0.0009) & 0.0405\,(0.0032) & 0.0114\,(0.0013) & 0.0147\,(0.0030) & 0.0125\,(0.0013) \\
$(20,4)$ & $\sigma=2$ & 1000 & 30 & \textbf{0.0232}\,(0.0035) & 0.0620\,(0.0100) & 0.0314\,(0.0059) & 0.0433\,(0.0108) & 0.0339\,(0.0056) \\
$(20,4)$ & $\kappa=0.01$ & 1000 & 30 & \textbf{0.0090}\,(0.0009) & 0.0426\,(0.0036) & 0.0128\,(0.0022) & 0.0389\,(0.0128) & 0.0135\,(0.0013) \\
$(20,4)$ & $\kappa=0.05$ & 1000 & 30 & \textbf{0.0094}\,(0.0012) & 0.0413\,(0.0039) & 0.0110\,(0.0017) & 0.0263\,(0.0048) & 0.0127\,(0.0013) \\
$(20,4)$ & $\kappa=0.2$ & 1000 & 30 & \textbf{0.0093}\,(0.0013) & 0.0408\,(0.0034) & 0.0150\,(0.0022) & 0.0134\,(0.0035) & 0.0128\,(0.0020) \\
$(20,4)$ & $\kappa=0.5$ & 1000 & 30 & \textbf{0.0108}\,(0.0012) & 0.0425\,(0.0028) & 0.0138\,(0.0021) & 0.0131\,(0.0020) & 0.0147\,(0.0022) \\
$(20,4)$ & Add $x_{5:8}$ & 1000 & 30 & \textbf{0.0094}\,(0.0009) & 0.0405\,(0.0032) & 0.0114\,(0.0013) & 0.0147\,(0.0030) & 0.0125\,(0.0013) \\
$(20,4)$ & $x_{5:8}$ only & 1000 & 30 & \textbf{0.0065}\,(0.0006) & 0.0405\,(0.0032) & 0.0114\,(0.0013) & 0.0147\,(0.0030) & 0.0125\,(0.0013) \\
\bottomrule
\end{tabular}
\end{table}
\clearpage

\subsection{Semi-real Data (RealCause)}\label{sec:realcause-experiments}
We use 50 realizations of the LaLonde CPS benchmark \citep{lalonde1986evaluating,dehejia1999causal} generated by RealCause \citep{neal2020realcause}. Each has 16,177 observations and eight covariates. All methods use only factual $(X,A,Y)$ for fitting, with a common $50/50$ train/test split (seed $50$; 8,088 training rows). Outcome scaling is fitted on training data; in realization $r$, effects are divided by the training outcome range $S_r$. The nuisance and target Mat\'ern kernels use $(\nu,\ell)=(1.5,1.3)$ and $(2.5,1.2)$, respectively, with the ridge grid and fitting schemes above. True propensities are unavailable, so oracle DR is omitted.

All methods use Nystr\"om KRR with at most 700 uniformly sampled training landmarks per fit (seed $700$); smaller fits use all training points. For $N$ fitting rows and $m$ landmarks, storage is $O(Nm+m^2)$ and ridge-solve cost is $O(Nm^2+m^3)$. Runs use macOS/ARM64 with single-threaded linear algebra per file; software settings and elapsed times are supplied with the code. The theoretical results concern exact KRR.

\paragraph{Evaluation.}
We compute test MSE between predicted effects and sampled differences $Y(1)-Y(0)$, both on the training-normalized outcome scale. This score includes a common outcome-noise contribution in addition to CATE estimation error.

We estimate the noise floor using the same 50 realizations. At each shared test covariate, we compute the sample variance of the unscaled effect across realizations, average over test covariates, and multiply by the average $S_r^{-2}$ to match the outcome normalization. Subtracting this common estimate from each method's mean test MSE gives its estimated excess risk and preserves method rankings. MSE SEs are computed across realizations. For excess-risk SEs, a joint delete-one-realization jackknife recomputes both the mean MSE and the correction after omitting each realization, accounting for their dependence.

The estimated noise floor is $0.08607$ (SE $0.00190$). Table~\ref{table:realcause} shows the lowest mean error for Ours; its estimated excess risk is $26.6\%$, $33.9\%$, and $85.7\%$ below plug-in, DR, and R, respectively.

\begin{table}[htbp]
\revisioncolor
\centering
\caption{RealCause results on 50 realizations. Entries are mean normalized error (SE). Estimated excess risk subtracts a common noise floor estimated from the same files; its SE includes correction uncertainty through the joint jackknife.}
\label{table:realcause}
\begin{tabular}{@{}lcc@{}}
\toprule
Method & Sampled-effect MSE (SE) & Estimated excess risk (SE) \\
\midrule
Ours & \textbf{0.10010} (0.00251) & \textbf{0.01402} (0.00125) \\
Plug-in & 0.10517 (0.00262) & 0.01909 (0.00114) \\
DR & 0.10730 (0.00267) & 0.02123 (0.00107) \\
R & 0.18427 (0.00420) & 0.09820 (0.00259) \\
\bottomrule
\end{tabular}
\end{table}
\clearpage
\end{revised}

\section{General Theory for the Analysis}\label{section: general theory}
We now turn to the technical core of the paper. The second stage of our method has the same form in Models~\labelcref{model; s1,model; s2,model; s3}: it is KRR applied to pseudo-outcomes produced by the nuisance stage, with only the second-stage function space changing across models.
This observation allows us to analyze all three cases through a single abstract theorem in a generic Hilbert space \(\XX\).
We first introduce the RKHS representation and generic notation, and then state the high-probability error bound from which all subsequent corollaries follow.

\subsection{RKHS Formulation and Notation}
\paragraph{RKHS Formulation.}
By standard RKHS theory \citep{aronszajn1950theory}, there exists a Hilbert space \(\mathbb{F}\) and a feature map \(\phi: \mathcal{X} \to \mathbb{F}\) satisfying the reproducing identity \(\langle \phi(x), \phi(x')\rangle_{\mathbb{F}} = K_{\mathcal{F}}(x,x')\).
Accordingly, the hypothesis space is defined as:
\begin{align*}
    \mathcal{F} &= \{ f_\theta(\cdot) = \langle \phi(\cdot), \theta \rangle_{\mathbb{F}} \mid \theta \in \mathbb{F} \}.
\end{align*}
We denote the Hilbertian element corresponding to \(f_\theta\) by \(\theta\).
In particular, \(\theta_0^\star\) and \(\theta_1^\star\) denote the Hilbertian elements corresponding to \(f_0^\star\) and \(f_1^\star\), respectively.

Under Model~\labelcref{model; s1}, the contrast class \(\mathcal{H}\) is itself an RKHS. Hence there exist a Hilbert space \(\mathbb{H}\) and a feature map \(\psi: \mathcal{X} \to \mathbb{H}\) such that \(\langle \psi(x), \psi(x')\rangle_{\mathbb{H}} = K_{\mathcal{H}}(x,x')\).
Accordingly, the hypothesis space is defined as:
\begin{align*}
    \mathcal{H} &= \{ h_\eta(\cdot) = \langle \psi(\cdot), \eta \rangle_{\mathbb{H}} \mid \eta \in \mathbb{H} \}.
\end{align*}
Similarly, we denote the Hilbertian element corresponding to \(h_\eta\) by \(\eta\).
We write \(\eta^\star\) for the Hilbertian element corresponding to \(h^\star\).

Under Model~\ref{model; s2}, we use the same notation \(\eta^\star\) for the Hilbertian element corresponding to \(h^\star\) in \(\mathbb{F}\).

Under Model~\labelcref{model; s3}, the space \(\tilde{\mathcal{H}}\) is again an RKHS, so there exist a Hilbert space \(\tilde{\mathbb{H}}\) and a feature map \(\tilde{\psi}: \mathcal{X} \to \tilde{\mathbb{H}}\) satisfying \(\langle \tilde{\psi}(x), \tilde{\psi}(x')\rangle_{\tilde{\mathbb{H}}} = K_{\tilde{\mathcal{H}}}(x,x')\).
Accordingly, the hypothesis space is defined as:
\begin{align*}
    \tilde{\mathcal{H}} &= \{ \tilde h_\eta(\cdot) = \langle \tilde{\psi}(\cdot), \eta \rangle_{\tilde{\mathbb{H}}} \mid \eta \in \tilde{\mathbb{H}} \}.
\end{align*}
We denote the Hilbertian element corresponding to \(\tilde{h}_\eta\) by \(\eta\).
We write \(\eta^\star\) for the Hilbertian element corresponding to \( h^\star\) in \(\tilde{\mathbb{H}}\).

\paragraph{Design Operators.}
For generic elements \(\{v_1, \dots, v_N\}\) in a Hilbert space \(\mathbb{X}\), define the \emph{design operator} \(\mathbf{V}: \mathbb{X} \to \mathbb{R}^N\) by
\begin{align*}
    \mathbf{V} \theta = (\langle v_1, \theta \rangle, \langle v_2, \theta \rangle, \dots, \langle v_N, \theta \rangle)^\top.
\end{align*}
Its adjoint \(\mathbf{V}^\top: \mathbb{R}^N \to \mathbb{X}\) is given by
\begin{align*}
    \mathbf{V}^\top \mathbf{a} = \sum_{i=1}^N a_i v_i \in \mathbb{X}.
\end{align*}

\subsection{Setup}
We now describe a generic KRR problem that covers the second stage of all three models.
Across Models~1--3, the second stage is always a KRR problem in which the scalar pseudo-outcome \(m_i\) plays the role of the response and the covariate enters only through a Hilbert-space feature \(w_i \in \XX\). The only object that changes across models is therefore the choice of second-stage Hilbert space \(\XX\) and the corresponding feature representation \(w_i\). This leads to the common estimator \(\hat\eta\) in \cref{equation; hat eta general theory}, which is the unified form analyzed below.
Fix a Hilbert space \(\XX\) and i.i.d.\ features \(\{w_i\}_{i=1}^n \subset \XX\), with associated design operator \(\Wb\). For example, in Model~\ref{model; s1} we take \(\XX=\HH\) and \(w_i=\psi(x_i)\); in Model~\ref{model; s2}, \(\XX=\mathbb{F}\) and \(w_i=\phi(x_i)\); and in Model~\ref{model; s3}, \(\XX=\tilde{\HH}\) and \(w_i=\tilde{\psi}(\tilde{x}_i)\).
Across Models~1--3, the regression oracle \(\Ocr\) applies KRR to the pseudo-outcome vector \(\Mb := (m_1, \dots, m_n)\); only the choice of \(\XX\) and \(w_i\) changes.
 
\paragraph{Regression Oracle \(\Ocr\) as KRR on \(\XX\).}
We abstract the second-stage regression in Algorithm~\ref{algorithm; main} as KRR in a generic Hilbert space \(\XX\), which will be instantiated as \(\HH\), \(\mathbb{F}\), or \(\tilde{\HH}\) depending on the model.
Write the population second moment operator as
\begin{align*}
    \bSigma := \EE[w_i \otimes w_i].
\end{align*}
This setup specializes to our models as follows:
\begin{itemize}
    \item Model~\ref{model; s1}: $\XX = \HH$, $w_i = \psi(x_i)$ (feature map of $K_\cH$).
    \item Model~\ref{model; s2}: $\XX= \mathbb{F}$, $w_i = \phi(x_i)$ (feature map of $K_\cF$).
    \item Model~\ref{model; s3}: $\XX = \tilde{\HH}$, $w_i = \tilde{\psi}(\tilde{x}_i)$ (feature map of $K_{\tilde{\cH}}$).
\end{itemize}
In all cases, the target function is represented by an element $\eta^\star \in \XX$ such that
\begin{align*}
    \langle w_i, \eta^\star \rangle_{\XX} = h^\star(x_i).
\end{align*}

Recall the pseudo-outcomes $m_i$ from \cref{equation: switch imputation} and their vector form $\Mb$.
Then the second-stage estimator returned by \(\Ocr\) can be written as
\begin{align}\label{equation; hat eta general theory}
    \hat{\eta} = (\Wb^\top \Wb + n \lambda \Ib)^{-1}\Wb^\top \Mb.
\end{align}
Here \(\hat{\eta}\) denotes the Hilbertian element corresponding to the function estimator \(\hat{h}\) returned by Algorithm~\ref{algorithm; main}.

\paragraph{Evaluation.}
To treat different loss functions simultaneously, we evaluate errors through a positive trace-class operator \(\bSigma_{\operatorname{ref}}\) on \(\XX\).
For any estimator \(\hat{\eta}\), define
\begin{align*}
    \cE_{\operatorname{ref}}(\hat{\eta}) := \| \hat{\eta}-\eta^\star \|_{\bSigma_{\operatorname{ref}}}^2
    \;=\; \langle \hat{\eta}-\eta^\star, \bSigma_{\operatorname{ref}}(\hat{\eta}-\eta^\star)\rangle_{\XX}.
\end{align*}
This formalism recovers the usual loss functions of interest. In Model~\ref{model; s1}, for example, \(\bSigma\) is the second-moment operator for \(K_\cH\); choosing \(\bSigma_{\operatorname{ref}}=\bSigma\) yields the \(L^2(\mathcal{P}_x)\)-error, while choosing \(\bSigma_{\operatorname{ref}}=\psi(x_0)\otimes \psi(x_0)\) yields the pointwise error at \(x_0\). The same interpretation applies with \(\phi(x_0)\) or \(\tilde\psi(\tilde x_0)\) under Models~\ref{model; s2} and~\ref{model; s3}.
For convenience, we summarize the two main cases:
\begin{itemize}
    \item $L^2$-error: $\bSigma_{\operatorname{ref}} = \bSigma$.
    \item Pointwise evaluation at $x_0$: $\bSigma_{\operatorname{ref}} = \psi(x_0)\otimes \psi(x_0)$ (or $\phi(x_0)\otimes \phi(x_0)$ / $\tilde\psi(\tilde x_0)\otimes \tilde\psi(\tilde x_0)$ in Models~\ref{model; s2}--\ref{model; s3}), so $\cE_{\operatorname{ref}}(\hat{\eta}) = |\hat h(x_0)-h^\star(x_0)|^2$.
\end{itemize}

Given \(\bSigma_{\operatorname{ref}}\), define
\begin{align*}
    \Sb_\lambda := \bSigma_{\operatorname{ref}}^{\frac{1}{2}} (\bSigma + \lambda \Ib)^{-1}\bSigma_{\operatorname{ref}}^{\frac{1}{2}}.
\end{align*}
This operator plays the role of an effective-dimension term under the evaluation metric induced by \(\bSigma_{\operatorname{ref}}\).

Additional notation used only in the appendix proofs is collected in Appendix~\ref{appendix; general notation} (Table~\ref{tab:notation-common}). Table~\ref{tab:notation-general} summarizes the notation used in this section.
\begin{table}[H]
\centering
\caption{Notation used in this section.}
\label{tab:notation-general}
\vspace{0.3cm}
\begin{tabular}{lp{0.72\linewidth}}
\toprule
\textbf{Symbol} & \textbf{Definition} \\
\midrule
\multicolumn{2}{l}{\textit{Spaces and feature maps}} \\
\(\mathbb{F}\) & Hilbert space associated with kernel \(K_\cF\) and feature map \(\phi:\cX\to\mathbb{F}\). \\
\(\HH\) & Hilbert space associated with kernel \(K_\cH\) (under Model~\ref{model; s1}). \\
\(\tilde{\mathbb{H}}\) & Hilbert space associated with \(K_{\tilde\cH}\) (under Model~\ref{model; s3}). \\
\(\phi,\psi,\tilde\psi\) & Feature maps for \(K_\cF, K_\cH, K_{\tilde\cH}\), respectively. \\
\midrule
\multicolumn{2}{l}{\textit{Operators and vectors}} \\
\(\theta_0^\star,\theta_1^\star\) & Hilbertian elements corresponding to nuisance functions \(f_0^\star,f_1^\star\) in \(\mathbb{F}\). \\
\(\eta^\star\) & Hilbertian element corresponding to the target contrast function \(h^\star\) in \(\XX\). \\
    \(w_i\) & Feature vector in \(\XX\) for the second-stage regression. \\
\(\Wb\) & design operator of \(\{w_i\}_{i=1}^n\). \\
\(\bSigma\) & population second moment \(\EE[w_i \otimes w_i]\). \\
\(\bSigma_{\operatorname{ref}}\) & evaluation operator defining the loss. \\
\(\Sb_\lambda\) & sandwich operator \(\bSigma_{\operatorname{ref}}^{1/2}(\bSigma+\lambda\Ib)^{-1}\bSigma_{\operatorname{ref}}^{1/2}\). \\
\(m_i\) & pseudo-outcome defined by switch imputation in \cref{equation: switch imputation}. \\
\(\Mb\) & pseudo-outcome vector \((m_i)_{i\in[n]}\). \\
\bottomrule
\end{tabular}
\end{table}

\subsection{Upper Bound Analysis}
We now state a single upper bound for the unified second-stage estimator introduced above. The theorem controls the evaluation error \(\cE_{\operatorname{ref}}\) in the generic Hilbert-space formulation and therefore applies simultaneously to Models~1--3 after the appropriate instantiation of \(\XX\), \(w_i\), and \(\bSigma_{\operatorname{ref}}\).

\begin{theorem}[General error bound]\label{theorem: general}
Suppose we run Algorithm~\ref{algorithm; main} with main regularizer \(\lambda \gtrsim \log n/n\) and nuisance regularizer \(\bar\lambda \asymp \log n /n\), and with \(\Ocr\) instantiated as KRR on a Hilbert space \(\XX\), under the assumptions in Section~\ref{section: assumptions}.
Let $\hat{h}$ be the estimator produced by the algorithm and let $\hat{\eta}$ be the corresponding Hilbertian element in $\XX$ (defined in \cref{equation; hat eta general theory}).
Then, with probability at least $1-n^{-10}$, we have
\begin{align*}
    \cE_{\operatorname{ref}}(\hat{\eta}) \lesssim \frac{1}{n\kappa} \Tr(\Sb_\lambda) + \lambda^2 \|\Sb_\lambda \|_{\operatorname{op}} \|(\bSigma + \lambda \Ib)^{-1/2} \eta^\star\|_\XX^2 + \frac{1}{n \kappa} \|\Sb_\lambda \|_{\operatorname{op}}(\| f_0^\star\|_\cF^2 + \|f_1^\star\|_\cF^2).
\end{align*}
\end{theorem}

\paragraph{Discussion.} The theorem recovers the standard KRR bias-variance form in $\XX$, up to an effective noise inflation by a factor \(1/\kappa\) caused by treatment imbalance.
In particular, the leading rate is governed entirely by the operators \(\bSigma\) and \(\bSigma_{\operatorname{ref}}\) on \(\XX\); the nuisance space does not alter the spectral bias-variance terms. The only explicit nuisance contribution is the final term, which is of order \(\tilde{\cO}((n\kappa)^{-1})\) under the bounded-norm assumption on \(f_0^\star\) and \(f_1^\star\), and is therefore mild relative to the nonparametric rates of interest.
Consequently, once the relevant spectral decay and trace bounds are plugged in, Theorem~\ref{theorem: general} immediately yields the $L^2$-error and pointwise rates for Models~\ref{model; s1}--\ref{model; s3}.

\paragraph{Instantiations.}
For $L^2$-error, $\Sb_\lambda = \bSigma^{\frac{1}{2}}(\bSigma+\lambda \Ib)^{-1}\bSigma^{\frac{1}{2}}$ and
$\Tr(\Sb_\lambda)=\sum_j \frac{\rho_j}{\rho_j+\lambda}$, where $\{\rho_j\}$ are eigenvalues of \(\bSigma\).
For point evaluation, $\Sb_\lambda$ is rank-one with
$\Tr(\Sb_\lambda)=\|\Sb_\lambda\|_{\op}=\langle \psi(x_0),(\bSigma+\lambda \Ib)^{-1}\psi(x_0)\rangle$ (or with $\phi$ / $\tilde\psi$ under Models~\ref{model; s2}--\ref{model; s3}), i.e., the leverage score.
In Sobolev RKHS $\cH = H^\gamma(\cX)$, Lemma~\ref{lemma: Sobolev 1} yields the bound
$\Tr(\Sb_\lambda)\lesssim \lambda^{-d/(2\gamma)}$ under Model~\ref{model; s1}.

\section{Proof of the Theorem~\ref{theorem: general} (General Theory)}\label{sec:general-proof}
\subsection{Notation}\label{appendix; general notation}
We collect additional notation used in the proof.
\begin{table}[H]
\centering
\caption{Common notation used in the appendix proofs.}
\label{tab:notation-common}
\begin{tabular}{lp{0.72\linewidth}}
\toprule
\textbf{Symbol} & \textbf{Definition} \\
\midrule
\multicolumn{2}{l}{\textit{Second moment operators}} \\
\(\bSigma_\cF\) & \(\EE[\phi(x_i)\phi(x_i)^\top]\). \\
\(\bSigma_{\cF,0}\) & \(\EE[\phi(x_i)\phi(x_i)^\top \bm 1(a_i=0)]\). \\
\(\bSigma_{\cF,1}\) & \(\EE[\phi(x_i)\phi(x_i)^\top \bm 1(a_i=1)]\). \\
\midrule
\multicolumn{2}{l}{\textit{RKHS elements}} \\
\(\theta_0^\star\) & Hilbertian element corresponding to \(f_0^\star\) in \(\mathbb{F}\). \\
\(\theta_1^\star\) & Hilbertian element corresponding to \(f_1^\star\) in \(\mathbb{F}\). \\
\midrule
\multicolumn{2}{l}{\textit{Design operators and vectors}} \\
\(\Xb\) & design operator of \(\{\phi(x_i)\}_{i=1}^n\). \\
\(\Xb_1\) & design operator of \(\{\phi(x_i)\bm{1}(a_i=1)\}_{i=1}^n\). \\
\(\Xb_0\) & design operator of \(\{\phi(x_i)\bm{1}(a_i=0)\}_{i=1}^n\). \\
\(\Yb_1\) & vector of \(\{y_i\bm{1}(a_i=1)\}_{i=1}^n\). \\
\(\Yb_0\) & vector of \(\{y_i\bm{1}(a_i=0)\}_{i=1}^n\). \\
\(\bm \varepsilon_1\) & vector of \(\{\varepsilon_i\bm{1}(a_i=1)\}_{i=1}^n\). \\
\(\bm \varepsilon_0\) & vector of \(\{\varepsilon_i\bm{1}(a_i=0)\}_{i=1}^n\). \\
\(m_i\) & pseudo-outcome defined by switch imputation in \cref{equation: switch imputation}. \\
\(\Mb\) & pseudo-outcome vector \((m_i)_{i\in[n]}\) with decomposition
\(\begin{aligned}
\Mb &= \Xb (\theta_1^\star - \theta_0^\star) + \bm\varepsilon'  + \Xb_0 (\hat \theta_1-\theta_1^\star) - \Xb_1 (\hat \theta_0-\theta_0^\star).
\end{aligned}\) \\
\(\bm\varepsilon'\) & vector with entries \(\varepsilon_i (-1)^{a_i+1}\), \(i\in[n]\). \\
\bottomrule
\end{tabular}
\end{table}
\subsection{Good Event}
By Assumption~\ref{assumption; positivity}, for every \(\lambda' \in \{\lambda,\bar \lambda\}\) we have
\[
n\bSigma_{\cF,1} \preceq \frac{1}{\kappa}(n\bSigma_{\cF,0} + \Ib), \qquad
n\bSigma_{\cF,0} \preceq \frac{1}{\kappa}(n\bSigma_{\cF,1} + \Ib).
\]
Since \(\lambda' \gtrsim \log n / n\), we may absorb the identity term into \(n\lambda' \Ib\). Consequently, there exists an absolute constant \(c>1\) such that
\begin{equation}\label{equation: good event 0-1}
\begin{aligned}
    &\frac{\kappa}{c}(n\bSigma_{\cF,0} + n\lambda' \Ib) \preceq n\bSigma_{\cF,1} + n\lambda' \Ib \preceq \frac{c}{\kappa}(n\bSigma_{\cF,0} + n\lambda' \Ib), \\
    &\frac{\kappa}{c}(n\bSigma_{\cF,0} + n\lambda' \Ib) \preceq n\bSigma_{\cF} + n\lambda' \Ib \preceq \frac{c}{\kappa}(n\bSigma_{\cF,0} + n\lambda' \Ib), \\
    &\frac{\kappa}{c}(n\bSigma_{\cF,1} + n\lambda' \Ib) \preceq n\bSigma_{\cF} + n\lambda' \Ib \preceq \frac{c}{\kappa}(n\bSigma_{\cF,1} + n\lambda' \Ib).
\end{aligned}
\end{equation}
Applying Lemma~\ref{lemma: trace class concentration bounded} and Remark~\ref{remark; useful observation}, for any \(\lambda' \in \{\lambda, \bar \lambda\}\) and with probability at least \(1-n^{-11}\), there exists an absolute constant \(c_2 >1\) such that
\begin{equation}\label{equation: good event 0-2}
\begin{aligned}
& \frac{1}{c_2}(n \bSigma_{\cF,1} + n\lambda' \Ib)  \preceq \Xb_1^\top \Xb_1 +  n\lambda' \Ib \preceq c_2(n \bSigma_{\cF,1} +  n\lambda' \Ib)\\
& \frac{1}{c_2}(n \bSigma_{\cF,0} +  n\lambda' \Ib ) \preceq \Xb_0^\top \Xb_0 +  n\lambda'\Ib\preceq c_2(n \bSigma_{\cF,0} + n\lambda'\Ib) \\
& \frac{1}{c_2}(n \bSigma_{\cF} +  n\lambda' \Ib) \preceq \Xb^\top \Xb +  n\lambda'\Ib\preceq c_2(n \bSigma_{\cF} + n\lambda' \Ib).
\end{aligned}    
\end{equation}
Moreover, for any $\lambda' \in \{\lambda, \bar\lambda\}$, with probability at least $1-n^{-11}$,
\begin{align}\label{equation: good event general 1}
   \frac{1}{c}( n\bSigma + n\lambda' \Ib) \preceq \Wb^\top \Wb + n\lambda' \Ib \preceq c( n\bSigma + n\lambda' \Ib)
\end{align}
for some absolute constant $c>1$.

Let \(\Ecr_0\) be the event on which the above concentration inequalities hold. Then
\begin{align*}
    \PP[\Ecr_0 ] \ge 1- 3n^{-11}.
\end{align*}
Combining \cref{equation: good event 0-1} and \cref{equation: good event 0-2}, on the event \(\Ecr_0\) we obtain
\begin{equation}\label{equation: good event 0-3}
    \begin{aligned}
    &\frac{\kappa}{c}(\Xb_0^\top \Xb_0 +  n\lambda' \Ib) \preceq  \Xb_1^\top \Xb_1 +  n\lambda' \Ib \preceq \frac{c}{\kappa}(\Xb_0^\top \Xb_0 +  n\lambda' \Ib) \\
       & \frac{\kappa}{c}(\Xb_0^\top \Xb_0 +  n\lambda' \Ib) \preceq  \Xb^\top \Xb +  n\lambda' \Ib \preceq \frac{c}{\kappa}(\Xb_0^\top \Xb_0 +  n\lambda' \Ib) \\
          &  \frac{\kappa}{c}(\Xb_1^\top \Xb_1 +  n\lambda' \Ib) \preceq  \Xb^\top \Xb +  n\lambda' \Ib \preceq \frac{c}{\kappa}(\Xb_1^\top \Xb_1 +  n\lambda' \Ib)
\end{aligned}
\end{equation}
for some constant \(c>1\).
\subsection{Proof}
We work on the event $\Ecr_0$, which controls $\Xb_0,\Xb_1,\Xb$ as defined above.

\paragraph{Error Decomposition.}

Recalling the pseudo-outcome decomposition,
\begin{align*}
\hat \eta = (\Wb^\top \Wb +n\lambda  \Ib)^{-1} \Wb^\top \big(\Wb \eta^\star + \bm\varepsilon' + \Xb_0 (\hat \theta_1-\theta_1^\star) -\Xb_1 (\hat \theta_0-\theta_0^\star)  \big).
\end{align*}
The estimation error can be written as
\begin{align*}
&\hat{\eta}-\eta^\star \\
&= (\Wb^\top \Wb + n\lambda \Ib)^{-1}
\big( \Wb^\top \Xb_0 (\hat \theta_1-\theta^\star_1 )
  \rev{- \Wb^\top \Xb_1(\hat\theta_0-\theta^\star_0)}
  - n\lambda \eta^\star\big) \\
&\quad +(\Wb^\top \Wb + n\lambda\Ib)^{-1} \Wb^\top \bm \varepsilon'.
\end{align*}
Accordingly, the error decomposes as
\begin{align*}
&\|\bSigma_{\operatorname{ref}}^{\frac{1}{2} }(\hat{\eta}-\eta^\star)\|_{\XX} \\
&\lesssim \|\bSigma_{\operatorname{ref}}^{\frac{1}{2} }(\Wb^\top \Wb + n\lambda \Ib)^{-1} \Wb^\top \Xb_1(\theta^\star_0 -\hat\theta_0) \|_\XX + \|\bSigma_{\operatorname{ref}}^{\frac{1}{2} } (\Wb^\top \Wb + n\lambda \Ib)^{-1} \Wb^\top \Xb_0 (\hat \theta_1-\theta^\star_1 ) \|_\XX  \\
&\quad + \|\bSigma_{\operatorname{ref}}^{\frac{1}{2} }(\Wb^\top \Wb + n\lambda \Ib)^{-1}n\lambda \eta^\star \|_\XX  + \| \bSigma_{\operatorname{ref}}^{\frac{1}{2} }(\Wb^\top \Wb + n\lambda \Ib)^{-1} \Wb^\top \bm\varepsilon'\|_{\XX}.
\end{align*}
We first bound the term
\begin{align*}
&\|\bSigma_{\operatorname{ref}}^{\frac{1}{2} }(\Wb^\top \Wb + n\lambda \Ib)^{-1} \Wb^\top \Xb_0(\hat\theta_1-\theta^\star_1) \|_\XX \\
&=  \|\bSigma_{\operatorname{ref}}^{\frac{1}{2} } (\Wb^\top \Wb + n\lambda \Ib)^{-1}\Wb^\top \Xb_0 (\Xb_1^\top \Xb_1 + n\bar\lambda \Ib)^{-1}(\Xb_1^\top \bm \varepsilon_1  - n\bar\lambda \theta_1^\star) \|_\XX \\
&\le  \|\bSigma_{\operatorname{ref}}^{\frac{1}{2} } (\Wb^\top \Wb + n\lambda \Ib)^{-1}\Wb^\top \Xb_0 (\Xb_1^\top \Xb_1 + n\bar\lambda \Ib)^{-1}\Xb_1^\top \bm \varepsilon_1  \|_\XX \\
&\quad + \|\bSigma_{\operatorname{ref}}^{\frac{1}{2} } (\Wb^\top \Wb + n\lambda \Ib)^{-1}\Wb^\top \Xb_0 (\Xb_1^\top \Xb_1 + n\bar\lambda \Ib)^{-1} n\bar\lambda \theta_1^\star \|_\XX
\end{align*}

\paragraph{Bias Term.}
Under the event \(\Ecr_0\), we have
\begin{align*}
\mathsf{B}_1 &:=  \| \bSigma_{\operatorname{ref}}^{\frac{1}{2}}(\Wb^\top \Wb + n\lambda \Ib)^{-1}\Wb^\top \Xb_0 (\Xb_1^\top \Xb_1 + n\bar\lambda \Ib)^{-1}(n\bar\lambda \theta_1^\star) \|_{\XX}\\
&\lesssim \| \bSigma_{\operatorname{ref}}^{\frac{1}{2}}(\Wb^\top \Wb + n\lambda \Ib)^{-\frac{1}{2}}\|_{\operatorname{op}} \| (\Wb^\top \Wb + n\lambda \Ib)^{-\frac{1}{2}} \Wb^\top \|_{\operatorname{op}}  \\
&\quad \times  \|  \Xb_0 (\Xb_1^\top \Xb_1 + n\bar\lambda \Ib)^{-\frac{1}{2}} \|_{\operatorname{op}} \|(\Xb_1^\top \Xb_1 + n\bar\lambda\Ib)^{-\frac{1}{2}}(n\bar\lambda\theta_1^\star) \|_{\mathbb{F}}\\
&\stackrel{(i)}{\lesssim} \| \bSigma_{\operatorname{ref}}^{\frac{1}{2}}(n \bSigma + n\lambda \Ib)^{-\frac{1}{2}}\|_{\operatorname{op}} \times  1 \times \frac{1}{\sqrt{\kappa}} \times \sqrt{n\bar\lambda} \| \theta_1^\star\|_{\mathbb{F}} \\
&\stackrel{}{\lesssim} \frac{1}{\sqrt{n}} \|\Sb_\lambda \|_{\operatorname{op}}^{\frac{1}{2}} \times   1 \times \frac{1}{\sqrt{\kappa}} \times \sqrt{n\bar\lambda} \| \theta_1^\star\|_{\mathbb{F}} \\
&\lesssim  \frac{1}{\sqrt{\kappa}} \|\Sb_\lambda \|_{\operatorname{op}}^{\frac{1}{2}} \sqrt{\bar\lambda} \| \theta_1^\star\|_{\mathbb{F}} \\
&\lesssim \frac{\sqrt{\log n}}{\sqrt{n \kappa}} \|\Sb_\lambda \|_{\operatorname{op}}^{\frac{1}{2}} \| \theta_1^\star\|_{\mathbb{F}}.
\end{align*}
In step (i), we used the definition of the good event \(\Ecr_0\) (specifically \cref{equation: good event 0-3} and \cref{equation: good event general 1}).

\paragraph{Variance Term.}
By the Hanson--Wright inequality, on the event \(\Ecr_0\), we obtain the following upper bound with probability at least \(1-n^{-11}\):
{\begin{align*}
&\mathsf{V}_1^2 \\
&:=  \|\bSigma_{\operatorname{ref}}^{\frac{1}{2}}(\Wb^\top \Wb + n\lambda \Ib)^{-1}\Wb^\top \Xb_0 (\Xb_1^\top \Xb_1 + n\bar\lambda \Ib)^{-1}\Xb_1^\top \bm \varepsilon_1\|_{\XX}^2 \\
&\lesssim \log n \cdot \Tr\Bigl(
\bSigma_{\operatorname{ref}}^{\frac{1}{2}}(\Wb^\top \Wb + n\lambda \Ib)^{-1}\Wb^\top \Xb_0
(\Xb_1^\top \Xb_1 + n\bar\lambda \Ib)^{-1}\Xb_1^\top \Xb_1 \\
&\qquad \times (\Xb_1^\top \Xb_1 + n\bar\lambda \Ib)^{-1}\Xb_0^\top \Wb
(\Wb^\top \Wb + n\lambda \Ib)^{-1}\bSigma_{\operatorname{ref}}^{\frac{1}{2}}
\Bigr) \\
&\lesssim \Tr\Bigl(
\bSigma_{\operatorname{ref}}^{\frac{1}{2}}(\Wb^\top \Wb + n\lambda \Ib)^{-1}\Wb^\top \Xb
(\Xb_1^\top \Xb_1 + n\bar\lambda \Ib)^{-1}\Xb^\top \Wb \\
&\qquad \times (\Wb^\top \Wb + n\lambda \Ib)^{-1}\bSigma_{\operatorname{ref}}^{\frac{1}{2}}
\Bigr)\log n \\
&\stackrel{(i)}{\lesssim} \frac{1}{\kappa}\Tr\Bigl(
\bSigma_{\operatorname{ref}}^{\frac{1}{2}}(\Wb^\top \Wb + n\lambda \Ib)^{-1}\Wb^\top \Xb
(\Xb^\top \Xb + n\bar\lambda \Ib)^{-1}\Xb^\top \Wb \\
&\qquad \times (\Wb^\top \Wb + n\lambda \Ib)^{-1}\bSigma_{\operatorname{ref}}^{\frac{1}{2}}
\Bigr)\log n \\
&\lesssim \frac{1}{\kappa}\Tr\Bigl(
\bSigma_{\operatorname{ref}}^{\frac{1}{2}}(\Wb^\top \Wb + n\lambda \Ib)^{-1}\Wb^\top \Wb
(\Wb^\top \Wb + n\lambda \Ib)^{-1}\bSigma_{\operatorname{ref}}^{\frac{1}{2}}
\Bigr)\log n\\
&\lesssim \frac{1}{\kappa} \Tr(\bSigma_{\operatorname{ref}}^{\frac{1}{2}}(\Wb^\top \Wb + n\lambda \Ib)^{-1}\bSigma_{\operatorname{ref}}^{\frac{1}{2}}) \log n \\
&\stackrel{(ii)}{\lesssim} \frac{\log n}{n\kappa} \Tr(\bSigma_{\operatorname{ref}} (\bSigma +\lambda \Ib)^{-1}) \\
&= \frac{\log n}{n\kappa} \Tr( \Sb_\lambda).
\end{align*}}
Here, in step (i) we used \cref{equation: good event 0-3}, and in step (ii) we used \cref{equation: good event general 1}.

The remaining propagated term is bounded similarly:
\begin{align*}
\|\bSigma_{\operatorname{ref}}^{\frac{1}{2} }(\Wb^\top \Wb + n\lambda \Ib)^{-1} \Wb^\top \Xb(\theta^\star_0 -\hat\theta_0) \|^2_\XX \lesssim \frac{\log n}{n\kappa} \Tr(\Sb_\lambda)+\frac{{\log n}}{{n \kappa}}\| \theta_0^\star\|_{\mathbb{F}}^2.
\end{align*}

\paragraph{Ridge Bias Term.}
Finally, we bound the ridge bias term:
\begin{align*}
\|\bSigma_{\operatorname{ref}}^{\frac{1}{2} }(\Wb^\top \Wb + n\lambda \Ib)^{-1}n\lambda \eta^\star \|_\XX &\lesssim \|\bSigma_{\operatorname{ref}}^{\frac{1}{2} }(\Wb^\top \Wb + n\lambda \Ib)^{-\frac{1}{2}} \|_{\operatorname{op}}\cdot \|(\Wb^\top \Wb + n\lambda \Ib)^{-\frac{1}{2}}n\lambda \eta^\star \|_\XX  \\
&\stackrel{(i)}{\lesssim} \|\bSigma_{\operatorname{ref}}^{\frac{1}{2} }(n \bSigma + n\lambda \Ib)^{-\frac{1}{2}} \|_{\operatorname{op}}\cdot \|(n\bSigma + n\lambda \Ib)^{-\frac{1}{2}}n\lambda \eta^\star \|_\XX  \\
&\lesssim \frac{1}{\sqrt{n}} \|\Sb_\lambda \|_{\operatorname{op}}^{\frac{1}{2}} \times \sqrt{n} \lambda  \times \|(\bSigma + \lambda\Ib)^{-\frac{1}{2}}\eta^\star \|_\XX \\
&= \lambda \|\Sb_\lambda \|_{\operatorname{op}}^{\frac{1}{2}} \|(\bSigma + \lambda\Ib)^{-\frac{1}{2}}\eta^\star \|_\XX
\end{align*}
where in step (i) we used \cref{equation: good event general 1}.

Combining the variance and bias bounds above with the decomposition yields the stated bound in Theorem~\ref{theorem: general}.

\section{Proofs of the Main Corollaries from General Theory}
We now show how each main corollary follows from Theorem~\ref{theorem: general} by specializing \(\bSigma_{\operatorname{ref}}, \bSigma\), and \(\XX\), bounding \(\|\Sb_\lambda\|_{\op}\) and \(\Tr(\Sb_\lambda)\), and choosing \(\lambda\). Polylogarithmic factors are absorbed into \(\lesssim\).
Throughout this section, \(w_i\), \(\XX\), \(\bSigma\), and \(\bSigma_{\operatorname{ref}}\) are as defined in Section~\ref{section: general theory} and Table~\ref{tab:notation-general}.

\subsection{Proof of Corollary~\ref{corollary; KRR} (Model~\ref{model; s1}, \texorpdfstring{$L^2$}{L2}-error)}\label{section: apdx proof L2M1}
Set $\bSigma_{\operatorname{ref}}=\bSigma_\cH$ and $\bSigma=\bSigma_\cH$, where \(\bSigma_\cH = \EE[\psi(x)\otimes \psi(x)]\). Let
$\bSigma_\cH=\sum_{j\ge1}\rho_{\cH,j}\,u_j\otimes u_j$ be the eigendecomposition. Then
\[
\Sb_\lambda=\bSigma_\cH^{1/2}(\bSigma_\cH+\lambda\Ib)^{-1}\bSigma_\cH^{1/2}
\]
has eigenvalues $\rho_{\cH,j}/(\rho_{\cH,j}+\lambda)$, so
\[
\|\Sb_\lambda\|_{\op}=\max_j\frac{\rho_{\cH,j}}{\rho_{\cH,j}+\lambda}\le 1,
\qquad
\Tr(\Sb_\lambda)=\sum_{j}\frac{\rho_{\cH,j}}{\rho_{\cH,j}+\lambda}.
\]
For the ridge bias, note that
\[
\lambda^2\|(\bSigma_\cH+\lambda \Ib)^{-1/2}\eta^\star\|_\HH^2
\le \lambda\|\eta^\star\|_\HH^2,
\]
so the bias term is $O(\lambda)$. Plugging these into Theorem~\ref{theorem: general} yields
\[
\mathcal{E}_{L^2}(\hat h) \lesssim \frac{1}{n\kappa}\sum_{j}\frac{\rho_{\cH,j}}{\rho_{\cH,j}+\lambda} + \lambda.
\]
If $\rho_{\cH,j}\lesssim j^{-2\ell_\cH}$, the effective-dimension bound
$\sum_j \frac{\rho_{\cH,j}}{\rho_{\cH,j}+\lambda}\lesssim \lambda^{-1/(2\ell_\cH)}$ holds
\citep{zhang2023optimality,ma2023optimally}, and choosing
$\lambda \asymp (n\kappa)^{-\frac{2\ell_\cH}{2\ell_\cH+1}}$ gives the rate in Corollary~\ref{corollary; KRR}. The finite-rank case follows by the bound $\sum_j \frac{\rho_{\cH,j}}{\rho_{\cH,j}+\lambda}\le D$.

\subsection{Proof of Corollary~\ref{corollary; source} (Model~\ref{model; s2}, \texorpdfstring{$L^2$}{L2}-error)}\label{section: apdx proof L2M2}
Set $\bSigma_{\operatorname{ref}}=\bSigma_\cF$ and $\bSigma=\bSigma_\cF$, where \(\bSigma_\cF = \EE[\phi(x)\otimes \phi(x)]\). Write
$\bSigma_\cF=\sum_{j\ge1}\rho_{\cF,j}\,u_j\otimes u_j$. Then
$\|\Sb_\lambda\|_{\op}\le 1$ and
$\Tr(\Sb_\lambda)=\sum_j \frac{\rho_{\cF,j}}{\rho_{\cF,j}+\lambda}$.

For the ridge bias, we use the source condition in its spectral form:
there exists $\gamma \in\mathbb{F}$ such that $\gamma=\bSigma_\cF^{-\nu/2}\eta^\star$.
Since $\bSigma_{\operatorname{ref}}=\bSigma_\cF$, we can sharpen the generic bound by
tracking the $\bSigma_\cF$ factor
\begin{align*}
    \lambda^2\|\bSigma_\cF^{1/2}(\bSigma_\cF+\lambda\Ib)^{-1}\eta^\star\|^2
 &\lesssim  \lambda^2 \|(\bSigma_\cF+\lambda\Ib)^{-\frac{1-\nu}{2}}(\bSigma_\cF+\lambda\Ib)^{-\frac{\nu}{2}}\eta^\star\|^2\\
 &\lesssim \lambda^{1+\nu}.
\end{align*}
Thus the ridge bias term is $O(\lambda^{1+\nu})$ under the source condition. Therefore,
\[
\mathcal{E}_{L^2}(\hat h) \lesssim \frac{1}{n\kappa}\sum_{j}\frac{\rho_{\cF,j}}{\rho_{\cF,j}+\lambda} + \lambda^{1+\nu}.
\]
If $\rho_{\cF,j}\lesssim j^{-2\ell_\cF}$, then
$\sum_j \frac{\rho_{\cF,j}}{\rho_{\cF,j}+\lambda}\lesssim \lambda^{-1/(2\ell_\cF)}$
\citep{ma2023optimally}, and choosing
$\lambda \asymp (n\kappa)^{-\frac{2\ell_\cF}{1+2\ell_\cF(1+\nu)}}$ yields the rate in Corollary~\ref{corollary; source}.

\subsection{Proof of Corollary~\ref{corollary; model3-L2} (Model~\ref{model; s3}, \texorpdfstring{$L^2$}{L2}-error)}\label{section: apdx proof L2M3}
Set $\bSigma_{\operatorname{ref}}=\bSigma_{\tilde{\cH}}$ and $\bSigma=\bSigma_{\tilde{\cH}}$, where \(\bSigma_{\tilde{\cH}} = \EE[\tilde\psi(x)\otimes \tilde\psi(x)]\), and write
$\bSigma_{\tilde{\cH}}=\sum_{j\ge1}\tilde\rho_j\,\tilde u_j\otimes \tilde u_j$.
Then $\|\Sb_\lambda\|_{\op}\le 1$ and
$\Tr(\Sb_\lambda)=\sum_j \frac{\tilde\rho_j}{\tilde\rho_j+\lambda}$.
As in Model~\ref{model; s1}, the ridge bias satisfies
$\lambda^2\|(\bSigma_{\tilde{\cH}}+\lambda\Ib)^{-1/2}\tilde\eta^\star\|^2\le \lambda\|\tilde\eta^\star\|_{\tilde{\HH}}^2$,
so it is $O(\lambda)$ when $\|\tilde\eta^\star\|_{\tilde{\HH}}$ is bounded. Hence,
\[
\mathcal{E}_{L^2}(\hat h) \lesssim \frac{1}{n\kappa}\sum_{j}\frac{\tilde\rho_j}{\tilde\rho_j+\lambda} + \lambda.
\]
If $\tilde\rho_j\lesssim j^{-2\tilde\ell}$, then
$\sum_j \frac{\tilde\rho_j}{\tilde\rho_j+\lambda}\lesssim \lambda^{-1/(2\tilde\ell)}$
\citep{ma2023optimally}, and choosing
$\lambda \asymp (n\kappa)^{-\frac{2\tilde\ell}{1+2\tilde\ell}}$ gives the rate in Corollary~\ref{corollary; model3-L2}.

\subsection{Proof of Corollary~\ref{corollary:PE} (Models~\ref{model; s1}--\ref{model; s2}, point evaluation)}\label{section: apdx proof PEM1}
For a fixed $x_0$, let $\bSigma_{\operatorname{ref}}=\psi(x_0)\otimes \psi(x_0)$ and $\bSigma=\bSigma_\cH$, where \(\bSigma_\cH = \EE[\psi(x)\otimes \psi(x)]\); then
\[
\Sb_\lambda=\bigl((\bSigma+\lambda\Ib)^{-1/2}\psi(x_0)\bigr)\otimes\bigl((\bSigma+\lambda\Ib)^{-1/2}\psi(x_0)\bigr)
\]
is rank-one with
\[
\|\Sb_\lambda\|_{\op}=\Tr(\Sb_\lambda)=\cQ_\lambda^2
:=\langle \psi(x_0),(\bSigma+\lambda\Ib)^{-1}\psi(x_0)\rangle,
\]
the standard leverage score. Theorem~\ref{theorem: general} yields
\[
\mathcal{E}_{x_0}(\hat h) \lesssim \left(\frac{1}{n\kappa}+\lambda\right)\cQ_\lambda^2.
\]
When $\cH=H^\gamma(\cX)$ with $\gamma>d/2$, Lemma~\ref{lemma: Sobolev 1} (Sobolev embedding / leverage score control) gives
$\cQ_\lambda^2\lesssim \lambda^{-d/(2\gamma)}$ \citep{tuo2024asymptotic}.
Taking $\lambda\asymp (n\kappa)^{-1}$ yields $\mathcal{E}_{x_0}(\hat h)\lesssim (n\kappa)^{-(2\gamma-d)/(2\gamma)}$, as in Corollary~\ref{corollary:PE}.

\subsection{Proof of Corollary~\ref{corollary:PE-model3} (Model~\ref{model; s3}, point evaluation)}\label{section: apdx proof PEM3}

Let $\bSigma_{\operatorname{ref}}=\tilde{\psi}(\tilde x_0)\otimes \tilde{\psi}(\tilde x_0)$ and $\bSigma=\bSigma_{\tilde{\cH}}$, where \(\bSigma_{\tilde{\cH}} = \EE[\tilde\psi(x)\otimes \tilde\psi(x)]\), so that
\[
\Sb_\lambda=\bigl((\bSigma_{\tilde{\cH}}+\lambda\Ib)^{-1/2}\tilde{\psi}(\tilde x_0)\bigr)\otimes\bigl((\bSigma_{\tilde{\cH}}+\lambda\Ib)^{-1/2}\tilde{\psi}(\tilde x_0)\bigr)
\]
and
$\|\Sb_\lambda\|_{\op}=\Tr(\Sb_\lambda)=\tilde{\cQ}_\lambda^2:=\langle \tilde{\psi}(\tilde x_0),(\bSigma_{\tilde{\cH}}+\lambda\Ib)^{-1}\tilde{\psi}(\tilde x_0)\rangle$.
Then
\[
\mathcal{E}_{x_0}(\hat h) \lesssim \left(\frac{1}{n\kappa}+\lambda\right)\tilde{\cQ}_\lambda^2.
\]
If $\tilde{\cH}=H^\beta(\tilde{\cX})$ with $\beta>\tilde d/2$, Sobolev embedding Lemma~\ref{lemma: Sobolev 1} yields
$\tilde{\cQ}_\lambda^2\lesssim \lambda^{-\tilde d/(2\beta)}$ \citep{tuo2024asymptotic}. Choosing $\lambda\asymp (n\kappa)^{-1}$ gives
$\mathcal{E}_{x_0}(\hat h)\lesssim (n\kappa)^{-(2\beta-\tilde d)/(2\beta)}$, matching Corollary~\ref{corollary:PE-model3}.

\begin{revised}
\section{Proof of Theorem~\ref{theorem: model selection} (Oracle Inequality)}\label{sec:selection-proof}
\subsection{Empirical target error and proxy-based selection}
For a candidate \(g\), define its \emph{true empirical error}
\[
R_3(g)=\frac1{n_3}\sum_{i=1}^{n_3}(g(x_{3i})-h^\star(x_{3i}))^2.
\]
This differs from the observable proxy loss minimized by Algorithm~\ref{algorithm; model selection}. The latter includes outcome noise and cannot be substituted for \(R_3\) in a concentration argument around population CATE error.

\begin{lemma}[Empirical target-error oracle inequality]\label{lemma: in-sample oracle inequality}
Under the assumptions of Theorem~\ref{theorem: model selection}, on the covariance concentration event used below, condition on \(\cD_1\) and on the covariates and treatments in \(\cD_2,\cD_3\). With conditional probability at least \(1-\delta\) over the remaining outcome noise,
\[
R_3(\bar h_{\rm ms})\leq C\min_{g\in\bar{\mathcal M}}R_3(g)
+C\left[\frac{\log n}{n_2\kappa}
+\left(\frac1{n_3}+\frac1{n_2\kappa}\right)\log(L/\delta)\right].
\]
\end{lemma}
The candidate set and \(B\) are fixed under this conditioning, but the proxy-training outcomes in \(\cD_2\) are not conditioned away. As shown below, the conditional proxy-noise vector has sub-Gaussian norm at most $C(1+\sqrt{n_3/(n_2\kappa)})$, and its squared mean bias is at most $C n_3\log n/(n_2\kappa)$. Lemma~\ref{lemma: loss model selection}, with \(\gamma=1\), then proves the display. Combining the design event with \(\delta=n^{-12}\) gives the required empirical comparison with overall failure probability of order \(n^{-11}\), after increasing constants.

\subsection{Population conversion and the role of \texorpdfstring{\(B\)}{B}}
Condition now only on \(\cD_1\) (and any independent data used to choose \(B\)). Put \(M=B+B^\star\), \(R(g)=\mathcal E_{L^2}(g)\), and \(Z_i(g)=(g(x_{3i})-h^\star(x_{3i}))^2\). For each truncated candidate,
\[
0\leq Z_i(g)\leq M^2,
\qquad \operatorname{Var}(Z_i(g))\leq\EE Z_i(g)^2\leq M^2R(g).
\]
Bernstein's inequality, a union bound over \(L\) candidates, and \(\sqrt{uv}\leq \eta u+C_\eta v\) give, simultaneously for all candidates,
\begin{align*}
|R_3(g)-R(g)|
&\leq \sqrt{2M^2R(g)t/n_3}+2M^2t/(3n_3)\\
&\leq \tfrac12 R(g)+C M^2t/n_3,
\qquad t=\log(2Ln^{12}),
\end{align*}
with probability at least \(1-n^{-12}\). In particular,
\[
R(g)\leq2R_3(g)+CM^2t/n_3,
\qquad R_3(g)\leq\tfrac32R(g)+CM^2t/n_3.
\]
Apply the first inequality to \(\bar h_{\rm ms}\), the empirical oracle lemma, and the second inequality to a population-optimal candidate. Since $n_1\asymp n_2\asymp n_3\asymp n$, their intersection yields \eqref{eq:selection-B-cost}. No independence between these events is needed. The variance bound proportional to \(R(g)\) is essential: a uniform \(O(M^2/\sqrt{n_3})\) deviation would not establish the stated fast remainder. If a random \(B\) fails to bound \(h^\star\) with probability \(\delta_B\), a further union bound adds \(\delta_B\).
\end{revised}

\subsection{Proof of Lemma~\ref{lemma: in-sample oracle inequality} (Oracle Inequality for Empirical \texorpdfstring{$L^2$}{L2}-Error)}
\begin{revised}
Write $n_j=|\mathcal D_j|$ for $j=1,2,3$, with $n=n_1+n_2+n_3$, and use $\tilde\lambda\asymp\log n/n_2$ as in \eqref{equation: tilde f}. Expectations below condition on $\mathcal D_1$ and the covariates and treatments in $\mathcal D_2,\mathcal D_3$, leaving their outcome noises random.
Let $\Xb_j$ have rows $\phi(x_{ji})^\top$, and let $\Xb_{j,a}$ have rows $\mathbf 1(a_{ji}=a)\phi(x_{ji})^\top$, for $i\in[n_j]$ and $j=2,3$. Write $\bm\varepsilon_{j,a}=(\mathbf 1(a_{ji}=a)\varepsilon_{ji})_{i=1}^{n_j}$ and $\bm\varepsilon'_3=\bm\varepsilon_{3,1}-\bm\varepsilon_{3,0}$.
Let $\Mb_3 := (\tilde m_{3i})_{i=1}^{n_3}$ and $\Mb_3^\star := (h^\star(x_{3i}))_{i=1}^{n_3}$.
Let $\tilde\theta_a$ denote the Hilbertian element corresponding to $\tilde f_a$ for $a\in\{0,1\}$.
\begin{align}
\Mb_3
 &= (\Xb_{3,0}\tilde\theta_1-\Xb_{3,0}\theta_0^\star-\bm\varepsilon_{3,0})
 +(\Xb_{3,1}\theta_1^\star+\bm\varepsilon_{3,1}-\Xb_{3,1}\tilde\theta_0) \nonumber\\
 &=\Xb_3(\theta_1^\star-\theta_0^\star)+\bm\varepsilon'_3
 +\Xb_{3,0}(\tilde\theta_1-\theta_1^\star)-\Xb_{3,1}(\tilde\theta_0-\theta_0^\star).
\end{align}
Thus $\Mb_3-\Mb_3^\star=\mathsf V_{\operatorname{ms}}+\mathsf B_{\operatorname{ms}}$, where
\begin{align*}
\mathsf V_{\operatorname{ms}}
&:=\bm\varepsilon'_3
 +\Xb_{3,0}(\Xb_{2,1}^\top\Xb_{2,1}+n_2\tilde\lambda\Ib)^{-1}\Xb_{2,1}^\top\bm\varepsilon_{2,1}\\
&\quad-\Xb_{3,1}(\Xb_{2,0}^\top\Xb_{2,0}+n_2\tilde\lambda\Ib)^{-1}\Xb_{2,0}^\top\bm\varepsilon_{2,0},\\
\mathsf B_{\operatorname{ms}}
&:=n_2\tilde\lambda\Xb_{3,1}(\Xb_{2,0}^\top\Xb_{2,0}+n_2\tilde\lambda\Ib)^{-1}\theta_0^\star\\
&\quad-n_2\tilde\lambda\Xb_{3,0}(\Xb_{2,1}^\top\Xb_{2,1}+n_2\tilde\lambda\Ib)^{-1}\theta_1^\star.
\end{align*}

\paragraph{Good Event.}
Apply Lemma~\ref{lemma: trace class concentration bounded} with $n_j$ observations on $\mathcal D_j$, for $j=2,3$, at $\lambda'=\tilde\lambda$. Since $n_2\asymp n_3\asymp n$, a sufficiently large constant in $\tilde\lambda\asymp\log n/n_2$ gives the following bounds with probability at least $1-n^{-11}$ for each comparison, for some constant $c_2>1$:
\begin{equation}\label{equation: good event ms-2}
\begin{aligned}
&\frac1{c_2}(n_2\bSigma_{\cF,1}+n_2\lambda'\Ib)
\preceq\Xb_{2,1}^\top\Xb_{2,1}+n_2\lambda'\Ib
\preceq c_2(n_2\bSigma_{\cF,1}+n_2\lambda'\Ib),\\
&\frac1{c_2}(n_2\bSigma_{\cF,0}+n_2\lambda'\Ib)
\preceq\Xb_{2,0}^\top\Xb_{2,0}+n_2\lambda'\Ib
\preceq c_2(n_2\bSigma_{\cF,0}+n_2\lambda'\Ib),\\
&\frac1{c_2}(n_2\bSigma_{\cF}+n_2\lambda'\Ib)
\preceq\Xb_2^\top\Xb_2+n_2\lambda'\Ib
\preceq c_2(n_2\bSigma_{\cF}+n_2\lambda'\Ib).
\end{aligned}
\end{equation}
\begin{equation}\label{equation: good event ms-3}
\begin{aligned}
&\frac1{c_2}(n_3\bSigma_{\cF,1}+n_3\lambda'\Ib)
\preceq\Xb_{3,1}^\top\Xb_{3,1}+n_3\lambda'\Ib
\preceq c_2(n_3\bSigma_{\cF,1}+n_3\lambda'\Ib),\\
&\frac1{c_2}(n_3\bSigma_{\cF,0}+n_3\lambda'\Ib)
\preceq\Xb_{3,0}^\top\Xb_{3,0}+n_3\lambda'\Ib
\preceq c_2(n_3\bSigma_{\cF,0}+n_3\lambda'\Ib),\\
&\frac1{c_2}(n_3\bSigma_{\cF}+n_3\lambda'\Ib)
\preceq\Xb_3^\top\Xb_3+n_3\lambda'\Ib
\preceq c_2(n_3\bSigma_{\cF}+n_3\lambda'\Ib).
\end{aligned}
\end{equation}
Let $\Ecr_{\operatorname{ms}}$ be the event on which these inequalities hold. Then
\[
\PP[\Ecr_{\operatorname{ms}}]\geq1-6n^{-11}.
\]
Combining \cref{equation: good event ms-2,equation: good event ms-3} with the operator-overlap assumption, on $\Ecr_{\operatorname{ms}}$ we obtain, for $a\in\{0,1\}$,
\begin{equation}\label{equation: good event cross}
\begin{aligned}
\frac{\kappa n_3}{c n_2}(\Xb_{2,a}^\top\Xb_{2,a}+n_2\lambda'\Ib)
&\preceq\Xb_{3,1-a}^\top\Xb_{3,1-a}+n_3\lambda'\Ib\\
&\preceq\frac{c n_3}{\kappa n_2}(\Xb_{2,a}^\top\Xb_{2,a}+n_2\lambda'\Ib),
\end{aligned}
\end{equation}
for some constant $c>1$.

To apply Lemma~\ref{lemma: loss model selection}, it suffices to bound
\[
\|\Mb_3-\EE[\Mb_3]\|_{\psi_2}=\|\mathsf V_{\operatorname{ms}}\|_{\psi_2},
\qquad
\|\EE[\Mb_3]-\Mb_3^\star\|_2=\|\mathsf B_{\operatorname{ms}}\|_2.
\]

\paragraph{Bounding $\|\mathsf V_{\operatorname{ms}}\|_{\psi_2}$.}
Observe that
\begin{align*}
\|\mathsf V_{\operatorname{ms}}\|_{\psi_2}
&\lesssim1+\|\Xb_{3,0}(\Xb_{2,1}^\top\Xb_{2,1}+n_2\tilde\lambda\Ib)^{-1}\Xb_{2,1}^\top\|_{\operatorname{op}}\\
&\quad+\|\Xb_{3,1}(\Xb_{2,0}^\top\Xb_{2,0}+n_2\tilde\lambda\Ib)^{-1}\Xb_{2,0}^\top\|_{\operatorname{op}}.
\end{align*}
We bound one term; the other is analogous:
\begin{align*}
&\|\Xb_{3,0}(\Xb_{2,1}^\top\Xb_{2,1}+n_2\tilde\lambda\Ib)^{-1}\Xb_{2,1}^\top\|_{\operatorname{op}}^2\\
&=\|\Xb_{2,1}(\Xb_{2,1}^\top\Xb_{2,1}+n_2\tilde\lambda\Ib)^{-1}
\Xb_{3,0}^\top\Xb_{3,0}
(\Xb_{2,1}^\top\Xb_{2,1}+n_2\tilde\lambda\Ib)^{-1}\Xb_{2,1}^\top\|_{\operatorname{op}}\\
&\stackrel{(i)}{\lesssim}\frac{n_3}{\kappa n_2}
\|\Xb_{2,1}(\Xb_{2,1}^\top\Xb_{2,1}+n_2\tilde\lambda\Ib)^{-1}\Xb_{2,1}^\top\|_{\operatorname{op}}\\
&\leq\frac{n_3}{\kappa n_2},
\end{align*}
where step (i) follows from \cref{equation: good event cross}. Hence
\[
\|\mathsf V_{\operatorname{ms}}\|_{\psi_2}
\lesssim1+\sqrt{\frac{n_3}{\kappa n_2}}.
\]

\paragraph{Bounding $\|\mathsf B_{\operatorname{ms}}\|_2$.}
We have
\begin{align*}
\|\mathsf B_{\operatorname{ms}}\|_2
&\leq n_2\tilde\lambda\|\Xb_{3,1}(\Xb_{2,0}^\top\Xb_{2,0}+n_2\tilde\lambda\Ib)^{-1}\theta_0^\star\|_2\\
&\quad+n_2\tilde\lambda\|\Xb_{3,0}(\Xb_{2,1}^\top\Xb_{2,1}+n_2\tilde\lambda\Ib)^{-1}\theta_1^\star\|_2.
\end{align*}
We bound the first term as
\begin{align*}
&n_2\tilde\lambda\|\Xb_{3,1}(\Xb_{2,0}^\top\Xb_{2,0}+n_2\tilde\lambda\Ib)^{-1}\theta_0^\star\|_2\\
&\leq n_2\tilde\lambda
\|\Xb_{3,1}(\Xb_{2,0}^\top\Xb_{2,0}+n_2\tilde\lambda\Ib)^{-1/2}\|_{\operatorname{op}}
\|(\Xb_{2,0}^\top\Xb_{2,0}+n_2\tilde\lambda\Ib)^{-1/2}\theta_0^\star\|_{\mathbb F}\\
&\stackrel{(i)}{\lesssim}n_2\tilde\lambda
\sqrt{\frac{n_3}{\kappa n_2}}\frac{\|\theta_0^\star\|_{\mathbb F}}{\sqrt{n_2\tilde\lambda}}
\lesssim\sqrt{\frac{n_3\tilde\lambda}{\kappa}},
\end{align*}
where step (i) follows from \cref{equation: good event cross} and the last inequality uses the bounded nuisance RKHS norm. The other term satisfies the same bound, so
\[
\frac{\|\mathsf B_{\operatorname{ms}}\|_2^2}{n_3}
\lesssim\frac{\tilde\lambda}{\kappa}
\asymp\frac{\log n}{n_2\kappa}.
\]

Applying Lemma~\ref{lemma: loss model selection} to the $n_3$ validation points with $\gamma=1$, the remainder is bounded by
\[
\frac{\|\mathsf B_{\operatorname{ms}}\|_2^2}{n_3}
+\frac{\|\mathsf V_{\operatorname{ms}}\|_{\psi_2}^2\log(L/\delta)}{n_3}
\lesssim\frac{\log n}{n_2\kappa}
+\left(\frac1{n_3}+\frac1{n_2\kappa}\right)\log(L/\delta).
\]
This proves Lemma~\ref{lemma: in-sample oracle inequality}.
\end{revised}

\begin{lemma}[Theorem 5.2 from \citealt{wang2026pseudo}]\label{lemma: loss model selection}
		Let $\left\{\boldsymbol{z}_i\right\}_{i=1}^n$ be deterministic elements in a set $\mathcal{Z}$, let $g^{\star}$ and $\left\{g_j\right\}_{j=1}^m$ be deterministic functions on $\mathcal{Z}$, and let $\widetilde{g}$ be a random function on $\mathcal{Z}$. Define
	\begin{align*}
		\mathcal{L}(g)=\frac{1}{n} \sum_{i=1}^n\left|g\left(\boldsymbol{z}_i\right)-g^{\star}\left(\boldsymbol{z}_i\right)\right|^2
	\end{align*}
	for any function $g$ on $\mathcal{Z}$. Assume that the random vector $\widetilde{\boldsymbol{y}}=\left(\widetilde{g}\left(\boldsymbol{z}_1\right), \widetilde{g}\left(\boldsymbol{z}_2\right), \cdots, \widetilde{g}\left(\boldsymbol{z}_n\right)\right)^{\top}$ satisfies $\|\widetilde{\boldsymbol{y}}-\mathbb{E} \widetilde{\boldsymbol{y}}\|_{\psi_2} \leq V<\infty$. Choose any
	\begin{align*}
			\widehat{j} \in \underset{j \in[m]}{\operatorname{argmin}}\left\{\frac{1}{n} \sum_{i=1}^n\left|g_j\left(\boldsymbol{z}_i\right)-\widetilde{g}\left(\boldsymbol{z}_i\right)\right|^2\right\}.
	\end{align*}
	
	There exists a universal constant $C$ such that for any $\delta \in(0,1]$, with probability at least $1-\delta$ we have
	\begin{align*}
			\mathcal{L}\left(g_{\widehat{j}}\right) \leq \inf _{\gamma>0}\left\{(1+\gamma) \min _{j \in[m]} \mathcal{L}\left(g_j\right)+C\left(1+\gamma^{-1}\right)\left(\mathcal{L}(\mathbb{E} \widetilde{g})+\frac{V^2 \log (m / \delta)}{n}\right)\right\}.
	\end{align*}
	
	Consequently,
	\begin{align*}
		\mathbb{E} \mathcal{L}\left(g_{\widehat{j}}\right) \leq \inf _{\gamma>0}\left\{(1+\gamma) \min _{j \in[m]} \mathcal{L}\left(g_j\right)+C\left(1+\gamma^{-1}\right)\left(\mathcal{L}(\mathbb{E} \widetilde{g})+\frac{V^2(1+\log m)}{n}\right)\right\}.
	\end{align*}    
\end{lemma}

\section{Technical Lemmas}

\begin{lemma}[Key Inequality for Point Evaluation: Sobolev Space]\label{lemma: Sobolev 1}
Let \(K\) be a Sobolev kernel on \(H^\beta(\cX)\) with feature map \(\phi(\cdot)\), and let \(\bSigma = \int_{x \in \cX} \phi(x)\phi(x)^\top \mathrm{d}x\). Suppose \(\cX \subset \RR^d\) and \(\beta>d/2\). For any \(x_0 \in \cX\) and any \(\lambda>0\), we have
\begin{align*}
\phi(x_0)^\top (\bSigma + \lambda \Ib)^{-1} \phi(x_0) \lesssim  \lambda^{-\frac{d}{2\beta}}.
\end{align*}
Hence,
\begin{align*}
\phi(x_0)\phi(x_0)^\top \lesssim \lambda^{-\frac{d}{2\beta}} (\bSigma + \lambda \Ib).
\end{align*}
\end{lemma}

\begin{proof}
Let $Q = \phi(x_0)^\top (\bSigma + \lambda \Ib)^{-1} \phi(x_0)$. Our goal is to find an upper bound for $Q$.

Let $\theta \in \HH$ be defined as $\theta = (\bSigma + \lambda \Ib)^{-1} \phi(x_0)$.
The quantity $Q$ can be written as an inner product in the RKHS $\HH$:
$$
Q = \inner{\phi(x_0)}{(\bSigma + \lambda \Ib)^{-1} \phi(x_0)}_{\HH} = \inner{\phi(x_0)}{\theta}_{\HH}.
$$
From the definition of $\theta$, we have $(\bSigma + \lambda \Ib)\theta = \phi(x_0)$. 
Taking the inner product with $\theta$ on both sides gives:
$$
\inner{(\bSigma + \lambda \Ib)\theta}{\theta}_{\HH} = \inner{\phi(x_0)}{\theta}_{\HH}.
$$
This simplifies to:
$$
\inner{\bSigma\theta}{\theta}_{\HH} + \lambda\inner{\theta}{\theta}_{\HH} = Q.
$$
Let $g_\theta(x) = \inner{\phi(x)}{\theta}_{\HH}$ be the function in the Sobolev space $H^\beta(\mathcal{X})$ corresponding to $\theta \in \HH$. By the reproducing property, $g_\theta(x_0) = \inner{\phi(x_0)}{\theta}_{\HH} = Q$.

The norms of $g_\theta$ can be related to the expression for $Q$:
\begin{enumerate}
\item The squared $H^\beta$ norm is: $\|g_\theta\|_{H^\beta(\mathcal{X})}^2 = \|\theta\|_{\HH}^2 = \inner{\theta}{\theta}_{\HH}$.
\item The squared $L^2$ norm is: $\|g_\theta\|_{L^2(\mathcal{X})}^2 = \int_{\mathcal{X}} |g_\theta(x)|^2 dx = \int_{\mathcal{X}} \inner{\phi(x)}{\theta}_{\HH}^2 dx \\ = \inner{\left(\int_{\mathcal{X}}\phi(x)\phi(x)^\top dx\right)\theta}{\theta}_{\HH} = \inner{\bSigma\theta}{\theta}_{\HH}$.
\end{enumerate}
Substituting these into the expression for $Q$, we get:
$$
\|g_\theta\|_{L^2(\mathcal{X})}^2 + \lambda \|g_\theta\|_{H^\beta(\mathcal{X})}^2 = Q.
$$
Since both norm terms are non-negative, we can establish bounds for each:
$$
\|g_\theta\|_{L^2(\mathcal{X})}^2 \le Q \implies \|g_\theta\|_{L^2(\mathcal{X})} \le \sqrt{Q}
$$
$$
\lambda \|g_\theta\|_{H^\beta(\mathcal{X})}^2 \le Q \implies \|g_\theta\|_{H^\beta(\mathcal{X})} \le \sqrt{\frac{Q}{\lambda}}
$$
Now apply the interpolation inequality to $g_\theta$. Let $q = \frac{d}{2\beta}$.
$$
Q = g_\theta(x_0) \le \|g_\theta\|_{L^\infty(\mathcal{X})} \le C \|g_\theta\|_{L^2(\mathcal{X})}^{1-q} \|g_\theta\|_{H^\beta(\mathcal{X})}^{q}.
$$
Substitute the bounds we found for the norms:
$$
Q \le C \left(\sqrt{Q}\right)^{1-q} \left(\sqrt{\frac{Q}{\lambda}}\right)^{q}
$$
and hence
$$
Q^{\frac{1}{2}} \le C \cdot \lambda^{-\frac{q}{2}}
$$
which completes the first inequality.
\begin{align*}
\phi(x_0)^\top (\bSigma_{} +\lambda \Ib)^{-1 } \phi(x_0) \lesssim \lambda^{-\frac{d}{2\beta}}.
\end{align*}
For the second inequality, let \(v\in \HH\). Then
\begin{align*}
\langle \phi(x_0)\phi(x_0)^\top v, v\rangle_\HH
&= |\langle \phi(x_0), v\rangle_\HH|^2 \\
&= \Big|\big\langle (\bSigma+\lambda \Ib)^{-1/2}\phi(x_0),(\bSigma+\lambda \Ib)^{1/2}v\big\rangle_\HH\Big|^2 \\
&\le \phi(x_0)^\top (\bSigma+\lambda \Ib)^{-1}\phi(x_0)\cdot \langle(\bSigma+\lambda \Ib)v,v\rangle_\HH \\
&\lesssim \lambda^{-d/(2\beta)} \langle(\bSigma+\lambda \Ib)v,v\rangle_\HH.
\end{align*}
Since this holds for all \(v\), we conclude
\[
\phi(x_0)\phi(x_0)^\top \lesssim \lambda^{-d/(2\beta)}(\bSigma+\lambda\Ib).
\]
\end{proof}

\begin{lemma}[Corollary E.1 from \citealt{wang2026pseudo}]\label{lemma: trace class concentration bounded}
Let $\left\{\boldsymbol{x}_i\right\}_{i=1}^n$ be i.i.d.\ random elements in a separable Hilbert space $\mathbb{H}$ with $\boldsymbol{\bSigma}:= \mathbb{E}\left(\boldsymbol{x}_i \otimes \boldsymbol{x}_i\right)$ trace class. Define $\hat{\boldsymbol{\bSigma}}=\frac{1}{n} \sum_{i=1}^n \boldsymbol{x}_i \otimes \boldsymbol{x}_i$. Choose any constant $\gamma \in(0,1)$ and define the event $\mathcal{A}=\{(1-\gamma)(\boldsymbol{\bSigma}+\lambda \boldsymbol{I}) \preceq \hat{\boldsymbol{\bSigma}}+\lambda \boldsymbol{I} \preceq(1+\gamma)(\boldsymbol{\bSigma}+\lambda \boldsymbol{I})\}$.
1. If $\left\|\boldsymbol{x}_i\right\|_{\mathbb{H}} \leq \xi$ holds almost surely for some constant $\xi$, then there exists a constant $C \geq 1$ determined by $\gamma$ such that $\mathbb{P}(\mathcal{A}) \geq 1-\delta$ holds so long as $\delta \in(0,1 / 14]$ and $\lambda \geq \frac{C \xi \log (n / \delta)}{n}$.    
Hence, with probability at least \(1-\delta\), we have
\begin{align*}
\frac{1}{C\xi}\left(\widehat \bSigma + \frac{\log(n/\delta)}{n} \Ib \right) \preceq \bSigma + \frac{\log(n/\delta)}{n} \Ib \preceq C\xi\left(\widehat \bSigma + \frac{\log(n/\delta)}{n} \Ib \right).
\end{align*}
\end{lemma}
\begin{remark}[A useful observation]\label{remark; useful observation}
For any two trace-class operators \(\Sb_1, \Sb_2\), suppose that
	\begin{align*}
		\Sb_1  \preceq c_1(\Sb_2 + c_2\log n \Ib)
	\end{align*}
	for some absolute constant \(c_1>1, c_2 >0\). 
	Then, for any \(c_3 >0\), for \(c = \max(1, c_2 / c_3)\), we have
	\begin{align*}
		\Sb_1 \preceq c(\Sb_2 + c_3 \log n \Ib).
	\end{align*}
\end{remark}
\begin{lemma}[Lemma E.1 from \citealt{wang2026pseudo}]\label{lemma: quadratic martingale}
Suppose that $\boldsymbol{x} \in \mathbb{R}^d$ is a zero-mean random vector with $\|\boldsymbol{x}\|_{\psi_2} \leq 1$. There exists a universal constant $C>0$ such that for any symmetric and positive semi-definite matrix $\boldsymbol{\Sigma} \in \mathbb{R}^{d \times d}$,
\begin{align*}
\mathbb{P}\left(\boldsymbol{x}^{\top} \boldsymbol{\bSigma} \boldsymbol{x} \leq C \operatorname{Tr}(\boldsymbol{\bSigma}) t\right) \geq 1-e^{-r(\boldsymbol{\bSigma}) t}, \quad \forall t \geq 1.
\end{align*}

Here $r(\boldsymbol{\bSigma})=\operatorname{Tr}(\boldsymbol{\bSigma}) /\|\boldsymbol{\bSigma}\|_2$ is the effective rank of $\boldsymbol{\bSigma}$.
\end{lemma}

\clearpage
\begin{revised}
\section{Nuisance Misspecification}\label{sec:misspecification}
We extend Theorem~\ref{theorem: general} by allowing approximate outcome models. Write
\[
f_a^\star=q_a^\star+r_a,\qquad q_a^\star\in\mathcal F,\quad a\in\{0,1\},
\]
where $q_a^\star$ has bounded RKHS norm and $r_a$ is the approximation residual. The actual contrast $h^\star=f_1^\star-f_0^\star$ retains its assumed structure under Model~1, 2, or~3; we do not require $q_1^\star-q_0^\star=h^\star$.

\subsection{Rates under Models 1--3}\label{sec:misspec-model-rates}
\begin{corollary}[Bounded nuisance misspecification]\label{cor:misspec-model-rates}
Suppose $\max_{a\in\{0,1\}}\|r_a\|_{L^\infty(\mathcal P_x)}\leq\varepsilon$. Under the remaining assumptions and regularization choices of Corollaries~\ref{corollary; KRR}--\ref{corollary; model3-L2}, the squared $L^2(\mathcal P_x)$-error satisfies, with probability at least $1-n^{-10}$,
\begin{equation}\label{eq:misspec-model-rates}
\mathcal E_{L^2}(\hat h)\lesssim\frac{\varepsilon^2}{\kappa}+
\begin{cases}
(n\kappa)^{-\frac{2\ell_\cH}{1+2\ell_\cH}}, & \text{Model 1},\\[5pt]
(n\kappa)^{-\frac{2\ell_\cF(1+\nu)}{1+2\ell_\cF(1+\nu)}}, & \text{Model 2},\\[5pt]
(n\kappa)^{-\frac{2\tilde\ell}{1+2\tilde\ell}}, & \text{Model 3}.
\end{cases}
\end{equation}
Here the spectral exponents are as in the corresponding corollaries, and logarithmic factors are suppressed. For Model~1 with exponential decay or rank at most $D$, the first term inside the cases is replaced by $(n\kappa)^{-1}$ or $D/(n\kappa)$, respectively.
\end{corollary}

\begin{proof}
For $L^2$-error, $\|\Sb_\lambda\|_{\operatorname{op}}\leq1$, while $n^{-1}(\|\bm r_0\|_2^2+\|\bm r_1\|_2^2)\leq2\varepsilon^2$. Proposition~\ref{prop:misspecification} below therefore contributes at most $\varepsilon^2/\kappa$, up to constants. The remaining terms follow the calculations in Appendices~\ref{section: apdx proof L2M1}--\ref{section: apdx proof L2M3}, including the sharper source-condition bias bound for Model~2.
\end{proof}

Thus a well-specified rate $(n\kappa)^{-\alpha}$ is preserved whenever $\varepsilon^2\lesssim\kappa(n\kappa)^{-\alpha}$.

\subsection{General bound and proof}\label{sec:misspec-general-bound}
We use the notation of Appendices~\ref{section: general theory} and~\ref{sec:general-proof}, with $h^\star(x)=\langle w(x),\eta^\star\rangle_\XX$ and $\bSigma_{\operatorname{ref}}=\bSigma$ for $L^2$-error, and keep the kernels and regularizers fixed in the following decomposition.

Let $\theta_{a,q}^\star\in\mathbb F$ represent $q_a^\star$. Define
\[
\bm r_a=(r_a(x_i))_{i\in[n]},\quad
\mathbf D_a=\operatorname{diag}\bigl(\bm 1(a_i=a)\bigr)_{i\in[n]},\quad
\mathbf A_a=(\Xb_a^\top\Xb_a+n\bar\lambda\Ib)^{-1},
\]
and let $\mathbf A_\lambda=(\Wb^\top\Wb+n\lambda\Ib)^{-1}$. The residual contribution to the pseudo-outcomes is
\begin{equation}\label{eq:misspec-residual}
\bm d=\mathbf D_1\bm r_0-\mathbf D_0\bm r_1
+\Xb_0\mathbf A_1\Xb_1^\top\bm r_1
-\Xb_1\mathbf A_0\Xb_0^\top\bm r_0.
\end{equation}
Write the bound from Theorem~\ref{theorem: general}, with $q_a^\star$ in the nuisance norms, as
\begin{align*}
\mathcal R_{n,\kappa}(q)
&:=\frac{1}{n\kappa}\Tr(\Sb_\lambda)
+\lambda^2\|\Sb_\lambda\|_{\operatorname{op}}
\|(\bSigma+\lambda\Ib)^{-1/2}\eta^\star\|_\XX^2\\
&\quad+\frac{\|\Sb_\lambda\|_{\operatorname{op}}}{n\kappa}
\bigl(\|q_0^\star\|_\cF^2+\|q_1^\star\|_\cF^2\bigr).
\end{align*}

\begin{proposition}[Approximate outcome models]\label{prop:misspecification}
Under the conditions of Theorem~\ref{theorem: general}, with nuisance well-specification replaced by the decomposition above, the estimator satisfies, with probability at least $1-n^{-10}$,
\begin{equation}\label{eq:misspec-bound}
\cE_{\operatorname{ref}}(\hat\eta)
\lesssim\mathcal R_{n,\kappa}(q)+\Delta_{\rm mis},\qquad
\Delta_{\rm mis}:=\|\bSigma_{\operatorname{ref}}^{1/2}\mathbf A_\lambda\Wb^\top\bm d\|_\XX^2,
\end{equation}
where logarithmic factors are suppressed as in the general bound. On the same event,
\begin{equation}\label{eq:misspec-energy}
\Delta_{\rm mis}\lesssim
\frac{\|\Sb_\lambda\|_{\operatorname{op}}}{n}\|\bm d\|_2^2
\lesssim\frac{\|\Sb_\lambda\|_{\operatorname{op}}}{\kappa}
\left(\frac{\|\bm r_0\|_2^2+\|\bm r_1\|_2^2}{n}\right).
\end{equation}
\end{proposition}

\begin{proof}
We follow the error decomposition in Appendix~\ref{sec:general-proof}. The first-stage estimator satisfies
\begin{align*}
\hat\theta_a-\theta_{a,q}^\star
&=\mathbf A_a\bigl(\Xb_a^\top\bm\varepsilon_a-n\bar\lambda\theta_{a,q}^\star\bigr)
+\mathbf A_a\Xb_a^\top\bm r_a\\
&=: (\hat\theta_{a,q}-\theta_{a,q}^\star)+\mathbf A_a\Xb_a^\top\bm r_a.
\end{align*}
Thus $\hat\theta_{a,q}$ is the KRR fit to the responses $q_a^\star(x_i)+\varepsilon_i$ in arm $a$. Substituting into the switch-imputation formula gives
\[
\Mb=\Wb\eta^\star+\bm\varepsilon'
+\Xb_0(\hat\theta_{1,q}-\theta_{1,q}^\star)
-\Xb_1(\hat\theta_{0,q}-\theta_{0,q}^\star)+\bm d.
\]
Consequently,
\begin{align*}
\hat\eta-\eta^\star
&=\mathbf A_\lambda\Bigl[
\Wb^\top\bigl\{\bm\varepsilon'
+\Xb_0(\hat\theta_{1,q}-\theta_{1,q}^\star)
-\Xb_1(\hat\theta_{0,q}-\theta_{0,q}^\star)\bigr\}
-n\lambda\eta^\star\Bigr]\\
&\quad+\mathbf A_\lambda\Wb^\top\bm d.
\end{align*}
The first line has exactly the nuisance-bias, centered-noise, and target-ridge terms bounded in the proof of Theorem~\ref{theorem: general}, with $\theta_{a,q}^\star$ in place of $\theta_a^\star$. Applying those bounds and the squared triangle inequality proves \eqref{eq:misspec-bound}. This argument only uses the representation of the actual contrast in $\XX$, not the identity $q_1^\star-q_0^\star=h^\star$.

To bound the new term, work on the same good event $\Ecr_0$. By \eqref{equation: good event general 1},
\begin{align*}
\|\bSigma_{\operatorname{ref}}^{1/2}\mathbf A_\lambda\Wb^\top\|_{\operatorname{op}}^2
&\leq\|\bSigma_{\operatorname{ref}}^{1/2}\mathbf A_\lambda\bSigma_{\operatorname{ref}}^{1/2}\|_{\operatorname{op}}
\lesssim\frac{\|\Sb_\lambda\|_{\operatorname{op}}}{n}.
\end{align*}
Moreover, \eqref{equation: good event 0-3} gives, for $a\in\{0,1\}$,
\[
\|\Xb_{1-a}\mathbf A_a\Xb_a^\top\|_{\operatorname{op}}
\leq\|\Xb_{1-a}\mathbf A_a^{1/2}\|_{\operatorname{op}}
\|\mathbf A_a^{1/2}\Xb_a^\top\|_{\operatorname{op}}
\lesssim\kappa^{-1/2}.
\]
Since $\|\mathbf D_a\|_{\operatorname{op}}\leq1$ and $\kappa\leq1$, equation~\eqref{eq:misspec-residual} implies
$\|\bm d\|_2^2\lesssim\kappa^{-1}(\|\bm r_0\|_2^2+\|\bm r_1\|_2^2)$.
Combining these inequalities proves \eqref{eq:misspec-energy}.
\end{proof}

\paragraph{$L^2$ control of the residuals.}
More generally, \eqref{eq:misspec-energy} uses only the empirical residual energy $\widehat\varepsilon_n^2:=n^{-1}\sum_i\{r_0(x_i)^2+r_1(x_i)^2\}$. If the fixed residual functions satisfy $\|r_0\|_{L^2(\mathcal P_x)}^2+\|r_1\|_{L^2(\mathcal P_x)}^2\leq\varepsilon^2$, Markov's inequality gives $\widehat\varepsilon_n^2\leq\varepsilon^2/\delta$ with probability at least $1-\delta$. Thus \eqref{eq:misspec-model-rates} holds with the additional term $\varepsilon^2/(\kappa\delta)$ and probability at least $1-n^{-10}-\delta$. The uniform residual bound in Corollary~\ref{cor:misspec-model-rates} avoids this extra factor.
\end{revised}

\begin{revised}
\section{Limitations}\label{sec:limitations}
Our theory concerns binary treatments and outcome-based imputation, without double robustness. Combining undersmoothed outcome regression with a DR estimator is left to future work. Approximate outcome models introduce an additional error term, quantified and proved in Appendix~\ref{sec:misspecification}; the contrast structure itself is still assumed to be correctly specified. The low-dimensional model assumes a supplied representation; learning an unknown projection remains open. Our model-selection guarantee is limited to $L^2$ loss and does not establish pointwise adaptation over unknown contrast classes. The selection guarantee also requires a valid upper bound $B$ on $\|h^\star\|_\infty$. Our experiments use a fixed nuisance regularization constant; we do not study noise-adaptive tuning. We leave extensions to noise with only finite moments to future work.
\end{revised}

\newpage
\input{checklist.tex}

\end{document}

%% file: checklist.tex
\section*{NeurIPS Paper Checklist}

\begin{enumerate}

\item {\bf Claims}
    \item[] Question: Do the main claims made in the abstract and introduction accurately reflect the paper's contributions and scope?
    \item[] Answer: \answerYes{} 
    \item[] Justification: The abstract and introduction state the main contributions and scope, and these are reflected in the theory and experiments.

\item {\bf Limitations}
    \item[] Question: Does the paper discuss the limitations of the work performed by the authors?
    \item[] Answer: \answerYes{} 
    \item[] Justification: \rev{Appendix~\ref{sec:limitations} discusses the limitations; Appendices~\ref{section; model selection procedure} and~\ref{sec:misspecification} give details on adaptation and nuisance misspecification.}

\item {\bf Theory assumptions and proofs}
    \item[] Question: For each theoretical result, does the paper provide the full set of assumptions and a complete (and correct) proof?
    \item[] Answer: \answerYes{} 
    \item[] Justification: The assumptions are stated in the setup and theorem statements, and full proofs are provided in the appendix.

    \item {\bf Experimental result reproducibility}
    \item[] Question: Does the paper fully disclose all the information needed to reproduce the main experimental results of the paper to the extent that it affects the main claims and/or conclusions of the paper (regardless of whether the code and data are provided or not)?
    \item[] Answer: \answerYes{} 
    \item[] Justification: The paper specifies the data-generating processes, baselines, kernels, tuning choices, splits, and evaluation metrics used in the experiments.

\item {\bf Open access to data and code}
    \item[] Question: Does the paper provide open access to the data and code, with sufficient instructions to faithfully reproduce the main experimental results, as described in supplemental material?
    \item[] Answer: \answerYes{} 
    \item[] Justification: We include an anonymized supplementary package with the experiment code, the RealCause CSV files used in the paper, and instructions to reproduce the main results. Synthetic data are generated by the provided scripts, and Appendix~\ref{section; experiments} documents the evaluation setup.

\item {\bf Experimental setting/details}
    \item[] Question: Does the paper specify all the training and test details (e.g., data splits, hyperparameters, how they were chosen, type of optimizer) necessary to understand the results?
    \item[] Answer: \answerYes{} 
    \item[] Justification: Appendix~\ref{section; experiments} describes the train/test protocol, hyperparameters, regularization grids, kernel choices, and benchmark details.

\item {\bf Experiment statistical significance}
    \item[] Question: Does the paper report error bars suitably and correctly defined or other appropriate information about the statistical significance of the experiments?
    \item[] Answer: \answerYes{} 
    \item[] Justification: The experimental tables report standard errors of the mean over repeated runs or dataset realizations, and the number of repetitions is stated.

\item {\bf Experiments compute resources}
    \item[] Question: For each experiment, does the paper provide sufficient information on the computer resources (type of compute workers, memory, time of execution) needed to reproduce the experiments?
    \item[] Answer: \rev{\ifshowrevisions\colorlet{orange}{blue}\fi\answerNo{}}
    \item[] Justification: \rev{Appendix~\ref{sec:realcause-experiments} reports the Nystr\"om implementation and CPU platform; the saved per-file results contain elapsed times, and the supplied manifests record software versions and thread settings. We did not log peak memory or a complete hardware resource profile for every experiment.}
    
\item {\bf Code of ethics}
    \item[] Question: Does the research conducted in the paper conform, in every respect, with the NeurIPS Code of Ethics \url{https://neurips.cc/public/EthicsGuidelines}?
    \item[] Answer: \answerYes{} 
    \item[] Justification: The work is methodological and uses standard synthetic and semi-synthetic benchmarks, and we are not aware of any deviation from the NeurIPS Code of Ethics.

\item {\bf Broader impacts}
    \item[] Question: Does the paper discuss both potential positive societal impacts and negative societal impacts of the work performed?
    \item[] Answer: \answerNA{} 
    \item[] Justification: There is no societal impact of the work performed.
    
\item {\bf Safeguards}
    \item[] Question: Does the paper describe safeguards that have been put in place for responsible release of data or models that have a high risk for misuse (e.g., pre-trained language models, image generators, or scraped datasets)?
    \item[] Answer: \answerNA{} 
    \item[] Justification: The paper does not release high-risk models or scraped datasets.

\item {\bf Licenses for existing assets}
    \item[] Question: Are the creators or original owners of assets (e.g., code, data, models), used in the paper, properly credited and are the license and terms of use explicitly mentioned and properly respected?
    \item[] Answer: \answerYes{} 
    \item[] Justification: Appendix~\ref{section; experiments} cites the RealCause benchmark, states that the official repository is MIT-licensed, and explains that the anonymized supplementary package includes the specific RealCause CSV files used in the experiments.

\item {\bf New assets}
    \item[] Question: Are new assets introduced in the paper well documented and is the documentation provided alongside the assets?
    \item[] Answer: \answerNA{} 
    \item[] Justification: The submission does not introduce or release new datasets, models, or code assets.

\item {\bf Crowdsourcing and research with human subjects}
    \item[] Question: For crowdsourcing experiments and research with human subjects, does the paper include the full text of instructions given to participants and screenshots, if applicable, as well as details about compensation (if any)? 
    \item[] Answer: \answerNA{} 
    \item[] Justification: The paper does not involve crowdsourcing or research with human subjects.

\item {\bf Institutional review board (IRB) approvals or equivalent for research with human subjects}
    \item[] Question: Does the paper describe potential risks incurred by study participants, whether such risks were disclosed to the subjects, and whether Institutional Review Board (IRB) approvals (or an equivalent approval/review based on the requirements of your country or institution) were obtained?
    \item[] Answer: \answerNA{} 
    \item[] Justification: The paper does not involve crowdsourcing or research with human subjects.

\item {\bf Declaration of LLM usage}
    \item[] Question: Does the paper describe the usage of LLMs if it is an important, original, or non-standard component of the core methods in this research? Note that if the LLM is used only for writing, editing, or formatting purposes and does \emph{not} impact the core methodology, scientific rigor, or originality of the research, declaration is not required.
    \item[] Answer: \answerNA{} 
    \item[] Justification: LLMs are not part of the core methodology or experiments in this work.

\end{enumerate}